\documentclass[lettersize,journal]{IEEEtran}
\usepackage{amsmath,amsfonts}
\usepackage{amsmath}
\usepackage{amsthm}
\usepackage{algorithm}
\usepackage{algpseudocode}
\usepackage{setspace}
\usepackage{tabularx}
\usepackage{array}
\usepackage{enumerate}
\usepackage[caption=false,font=normalsize,labelfont=sf,textfont=sf]{subfig}
\usepackage{textcomp}
\usepackage{stfloats}
\usepackage{url}
\usepackage{verbatim}
\usepackage{graphicx}
\usepackage{epstopdf}
\usepackage{subfig}
\usepackage{cite}
\usepackage{xcolor}
\usepackage{makecell}
\usepackage{multirow}
\usepackage{booktabs}
\usepackage{bm}
\usepackage{balance}
\usepackage[numbers,sort&compress]{natbib}
\usepackage{lipsum}
\usepackage{textcomp,mathcomp}
\usepackage{gensymb}
\usepackage{dsfont}
\allowdisplaybreaks

\newtheorem{theorem}{Theorem}

\let\Algorithm\algorithm
\renewcommand\algorithm[1][]{\Algorithm[#1]\setstretch{1.2}}

\newcommand{\Taxi}{{\mathchoice{}{}{\scriptscriptstyle}{}T}}
\newcommand{\Rap}{{\mathchoice{}{}{\scriptscriptstyle}{}R}}

\def\p(#1|#2){p\left(#1\,|\,#2\right)} 

\begin{document}

\title{Optimal Power Management for Large-Scale Battery Energy Storage Systems via Bayesian Inference}

\author{Amir Farakhor,~\IEEEmembership{Graduate Student Member,~IEEE}, Iman Askari,~\IEEEmembership{Graduate Student Member,~IEEE}, Di Wu,~\IEEEmembership{Senior Member,~IEEE}, Yebin Wang,~\IEEEmembership{Senior Member,~IEEE}, and Huazhen Fang,~\IEEEmembership{Member,~IEEE}
% <-this % stops a space
\thanks{A. Farakhor, I. Askari, and H. Fang (corresponding author) are with the Department of Mechanical Engineering, University of Kansas, Lawrence, KS, USA (Email: a.farakhor@ku.edu, askari@ku.edu, fang@ku.edu).} 
\thanks{D. Wu is with the Pacific Northwest National Laboratory, Richland, WA, USA (Email: di.wu@pnnl.gov).}
\thanks{Y. Wang is with the Mitsubishi Electric Research Laboratories, Cambridge, MA, USA (Email: yebinwang@merl.com).}
% <-this % stops a space
}
% The paper headers
\markboth{}
{Shell \MakeLowercase{\textit{et al.}}: A Sample Article Using IEEEtran.cls for IEEE Journals}

\maketitle

\begin{abstract}
Large-scale battery energy storage systems (BESS) have found ever-increasing use across industry and society to accelerate clean energy transition and improve energy supply reliability and resilience. However, their optimal power management poses significant challenges: the underlying high-dimensional nonlinear nonconvex optimization lacks computational tractability in real-world implementation, and the uncertainty of the exogenous power demand makes exact optimization difficult. This paper presents a new solution framework to address these bottlenecks. The solution pivots on introducing  {\em power-sharing ratios} to specify each cell's power quota from the output power demand. To find the optimal power-sharing ratios, we formulate a nonlinear model predictive control (NMPC) problem to achieve power-loss-minimizing BESS operation while complying with safety, cell balancing, and power supply-demand constraints. We then propose a parameterized control policy for the power-sharing ratios, which utilizes only three parameters, to reduce the computational demand in solving the NMPC problem. This policy parameterization allows us to translate the NMPC problem into a Bayesian inference problem for the sake of 1) computational tractability, and 2) overcoming the nonconvexity of the optimization problem. We leverage the ensemble Kalman inversion technique to solve the parameter estimation problem. Concurrently, a low-level control loop is developed to seamlessly integrate our proposed approach with the BESS to ensure practical implementation. This low-level controller receives the optimal power-sharing ratios, generates output power references for the cells, and maintains a balance between power supply and demand despite uncertainty in output power. We conduct extensive simulations and experiments on a 20-cell prototype to validate the proposed approach.
\end{abstract}


\section{Introduction}
\IEEEPARstart{B}{attery} energy storage systems (BESS) have been rising as an enabler for electrified transportation, renewable energy, and grid transformation, by harnessing the strengths of lithium-ion batteries, including fast charging/discharging, high energy efficiency, and long cycle life \cite{TTE-ZW-2021, 2022-ITEC-SM, TTE-LY-2022}. In the upcoming decades, they will become an essential part of our energy infrastructure vital for future sustainability and decarbonization. BESS assembles a large number of lithium-ion battery cells into a hierarchy of modules and packs wired together to provide high power and large capacity. It requires optimal  power management (OPM) to maximize the performance from cell to system level \cite{TTE-WY-2020, 2019-EnergyResearch-SM}. At the core of this problem is allocating the total charging/discharging power among the constituent units with an awareness of safety, cell life and balancing, and energy loss, among others. Due to its importance, the OPM problem has attracted a growing body of work for a diversity of applications from utility-scale energy storage to electric aircraft \cite{2020-arXiv-BA, EnergyConversion-AB-2022}. Recent research advances, however, still leave two bottlenecks open. First, OPM entails high-dimensional nonlinear nonconvex optimization, which is computationally prohibitive for large-size BESS. Second, existing studies generally assume perfect knowledge of future output power demands to facilitate OPM design. This assumption nonetheless is unrealistic, and uncertain output power will considerably complicate OPM design, with little known yet about how to handle it. In this paper, we address these two bottlenecks both at once.  Specifically, we introduce the concept of power-sharing ratios (PSR), which refer to each storage unit's relative share to contribute in charging/discharging; further, we propose a Bayesian inference approach to efficiently identify the optimal PSR based on the output power demands.
\subsection{Literature Review}
We will first review the literature on OPM for BESS based on perfect future power demand prediction. Further, although few existing have taken the future power demand uncertainty into account, we will survey uncertain power management for microgrids, which is a related line of research. Finally, we will review the literature on the intersection of control and Bayesian inference for its connection with the  OPM approach to be proposed.

\subsubsection{OPM with perfect power demand knowledge}
Today's BESS increasingly integrate power electronics circuits with batteries for different functions, including balancing and power flow distribution. These circuits make sophisticated OPM possible to extract better performance from BESS. In this context, the study in \cite{2014-VPPC-BJ} formulates and addresses a multi-objective control problem to achieve optimal power allocation with both charge and temperature balancing. Extending this work, some further studies, e.g., \cite{2016-TSTE-PC}, focus on setting up more solvable convex optimization problems. To enhance BESS performance, a promising avenue is concurrent circuit and OPM design. For example, the BESS design in \cite{2019-TVT-CR} embed power converters and supercapactiors into the balancing circuit and leverage optimal control to attain charge and temperature equalization as well as energy loss minimization. As shown in \cite{2023-TTE-FA}, reconfigurable power electronics architectures based on converters and switches can make BESS safer and more amenable to power regulation, especially when combined with algorithms to perform simultaneous circuit topology and power control. In these studies, the OPM approaches entail significant computational requirements, especially for large BESS configurations. To address this issue, \cite{2021-TCST-CR} presents a multi-layer control framework for BESS, which divides control tasks into different layers with different time scales and model complexities to reduce computational costs. Distributing computation among the constituent units of BESS allows every unit to handle a limited amount of computation, thus capable of accelerating the overall computational speed, as shown by the distributed OPM algorithms in \cite{2023-ACC-FA, 2018-TSTE-OQ}. Another useful approach to mitigate the computational burden is  clustering the batteries within a BESS into different groups according to their characteristics and then performing cluster-based optimization, which is much smaller in scale \cite{2024-TTE-FA}. While these studies represent significant progress, OPM for BESS still faces high computational complexity. This complexity fundamentally arises from the large number of optimization variables, which increases with the size of the BESS, in addition to the nonlinear nonconvex nature of the problem.

Meanwhile, current OPM approaches rely on the impractical assumption of perfect knowledge regarding output power demand, holding them back from reaching performance levels intended by design in real-world scenarios. Important as it is, the consideration of uncertainty in output demand has been largely unexplored. The study in \cite{2017-TCST-AF} proposes to use short-term power load predictions, rather than more varying long-term predictions, for OPM, while sacrificing the power management performance. Other than this work, the consideration of uncertain output power demand has been rare in the literature, highlighting the need for more substantial research.
\subsubsection{Microgrid energy management}
In relation to uncertain OPM for BESS, it is interesting to glance at microgrid power management. This problem centers around optimal coordination and power scheduling for a microgrid comprised of diverse energy resources, including generators, storages, and renewables \cite{2018-TCST-RC}. Uncertainty spreads within the microgrid, especially for renewable generation and loads \cite{2022-IEEEAccess-MG, 2022-EnergySystems-DA}. This has motivated considerable research in load and renewable forecasting \cite{2022-AppliedEnergy-JZ}, but   forecasting   cannot eliminate the uncertainty. Many studies thus pursue the development of robust microgrid power management approaches. Min-max optimization has found its way in \cite{2016-SmartGrid-VF, 2021-Automation-HS} to ensure the microgrid performance in the worst case or in some specific conservative sense. Stochastic optimization and Lyapunov optimization have also shown effective when stochastic uncertainty impacts a microgrid \cite{2013-RenewableEnergy-AB, 2022-IEEEAccess-ME, 2021-AppliedEnergy-Amin}. While these  techniques can suppress the effects of uncertainty, they impose heavy computational costs. Applying them to BESS is almost impractical from a computational perspective, as a BESS has significantly more constituent units than a typical microgrid.
\subsubsection{Optimal control by inference}
The OPM problem has been generally studied within  optimal control frameworks, involving constrained nonconvex optimization~\cite{2012-IFAC-NM, 2016-TSTE-PC, 2019-TVT-CR,2023-TTE-FA,2021-TCST-CR}. The literature has commonly adopted gradient-based optimization solvers to address this problem. These solvers, however, require heavy computation, struggle to overcome nonconvexity, and  fail easily for large-scale BESS. Some gradient-free global optimization solvers may provide an alternative, but they also pose severe computational burden~\cite{2005-NaturalComputation-WX}. The status quo motivates to develop  Bayesian inference-based OPM approaches, which promise more efficient computation as needed for practical implementation. 

While never explored for OPM, the idea of treating optimal control via inference has been around for many years. The earliest and most well-known result is the duality between linear quadratic regulator and Kalman smoothing~\cite{Kalman:JBE:1960}. Since then, the control/estimation duality has been a topic of recurring interest~\cite{2008-CDC-TE,Kim:ACC:2024}. Influenced by duality, control design by inference has gained momentum. Some studies have reformulated stochastic optimal control as variational inference problems~\cite{2012-ML-KH,Williams:TRO:2018}, with the hope of tackling complicated problems unyielding to optimization. Bayesian inference has found increasing use in dealing with NMPC. An early study in this regard is~\cite{Toussaint:ICML:2009}, which uses approximate Bayesian inference and message passing to deal with stochastic optimal control. Similar studies follow in~\cite{2020-PMLR-JW, 2023-ACC-SS}. Particle filtering is used  to solve NMPC in~\cite{2011-ControlLetters-DS}, and the works in~\cite{2021-ACC-AI, Asakri:ACC:2022} further develop constraint-aware particle filtering to address constrained NMPC. If designed to achieve high-quality sampling, particle filtering can provide considerable computational efficiency for complex NMPC problems like those involving neural network dynamic models~\cite{2023-Arxiv-IA}. Related with optimal control, reinforcement learning has seen some solutions rooted in Bayesian inference~\cite{Levine:arXiv:2018}.  From the literature, Bayesian inference has a strong capability of handling optimal control for complex systems, which will be leveraged to enable computationally fast OPM for large-scale BESS.

\subsection{Statement of Contributions}
As reviewed above, the study of OPM for BESS has recently gained significant momentum, but a critical question remains: how can we enable OPM with fast computation while also handling uncertain output power demands? This question has been acquiring more urgency as today's BESS are rapidly increasing in size and facing greater operational uncertainty.

In this paper, we present a new framework solution to tackle the issue. Our study pivots on a key observation: for a storage unit within a BESS, its power allocation is tied to the unit's operating condition, e.g., state-of-charge (SoC), temperature, and internal resistance. This observation inspires us to propose the concept of {\em power sharing ratio} (PSR), which stands for the power quota of the unit relative to the predicted total output power demand. The PSR can be parameterized as a function of the unit's operating condition. Finding the best parameters of the PSR function then is tantamount to solving the OPM problem---this effectively translates the OPM problem into a parameter inference or estimation problem. Using the PSR as the decision variables in OPM will present two major advantages. 
\begin{enumerate}[1)]
	\item The parameterized PSR function includes only several parameters to identify, even for a large-size BESS, contrasting with thousands of decision variables to optimize in traditional OPM methods. This makes it possible to reduce the amounts of computation by considerable margins. 
	\item The PSR measures a storage unit's power-handling capability relative to other units, based on the unit's operating condition. As such, it is a reliable metric for splitting the total power demand under varying demand scenarios and thus capable of managing uncertain power demands.
\end{enumerate}

Following the above, we develop the following specific contributions to build a holistic PSR-based OPM framework.

\begin{enumerate}[1)]
	\item We introduce the definition of PSR and establish a two-level PSR-based OPM architecture. At the top level is an NMPC for PSR-based power allocation. At the lower level is a proportional-integral (PI) controller, which leverages the optimal PSR from the top level to determine and track each unit's output power reference. The architecture will ensure both near-optimal power allocation and the power supply-demand balance under imperfect output power predictions. 
	\item We propose a parameterized form of the PSR, which captures the PSR's dependence on the unit's SoC, temperature and internal resistance. The  NMPC-based OPM problem then drills down to a problem of estimating the best parameters of the PSR function using the control objectives and constraints as the evidence. To address this non-trivial parameter estimation problem, we adopt the ensemble Kalman inversion (EnKI) approach, which combines effective sampling with sequential computation to achieve high computational efficiency.
	\item We develop a 20-cell testbed to assess the performance of the proposed approach. The experiments involve the operation of the BESS under output power uncertainty and unbalanced cell conditions. The obtained results demonstrate that the proposed OPM approach is computationally fast and scalable while robust against uncertainty in power demand.
\end{enumerate}

A preliminary version of this study was presented in \cite{ACC-FA-2024}. The current work significantly expands upon it by 1) providing a more complete problem formulation, 2) introducing the two-level OPM architecture that incorporates a low-level PI controller for the purpose of practical implementation, and 3) adding the experimental validation results based on a lab-scale prototype.
\subsection{Organization}
The rest of the paper is organized as follows. Section II outlines the circuit structure and OPM control architecture of the considered BESS. Section III presents the electro-thermal modeling of the considered BESS and subsequently formulates an NMPC problem for OPM. Section VI expresses the NMPC problem as a Bayesian inference problem. Section V delves into the design of the low-level controller. Sections VI and VII provide the simulation and experimental results, respectively, to validate the proposed approach. Finally, Section VIII concludes the paper with final remarks and insights.
\section{BESS Circuit Structure and Control Architecture}
This section introduces the circuit structure of the considered BESS, and then describes the control problem and architecture for its OPM. 

\begin{figure}[!t]
\centering
\includegraphics[width=\linewidth]{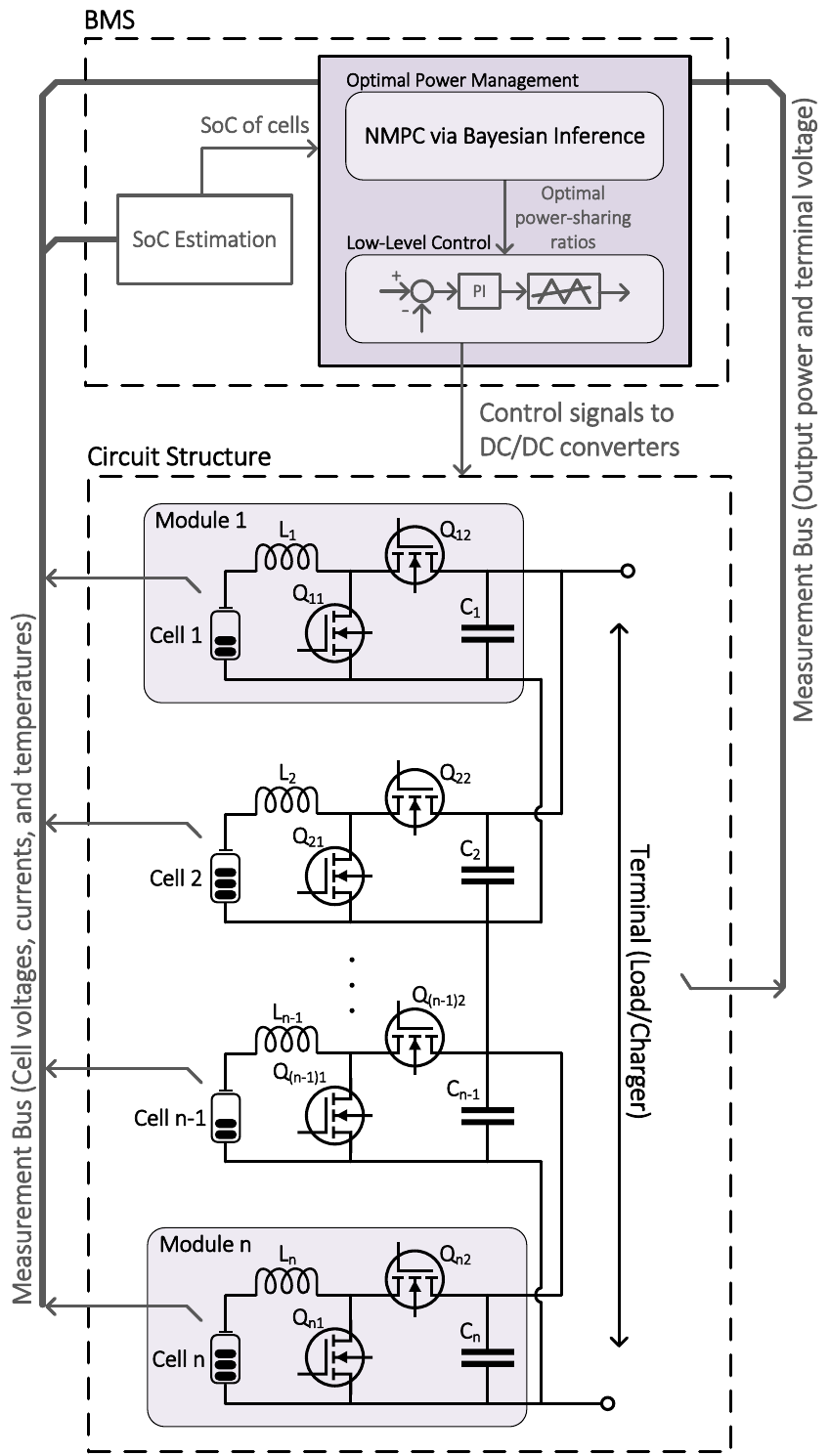}
\caption{The BESS circuit structure and control architecture.}
\label{FIG-1}
\end{figure}

Fig.~\ref{FIG-1} illustrates the circuit structure of the considered BESS and its interaction with the OPM approach. The BESS comprises $n$ modules, each module consisting of a storage unit connected to a DC/DC converter. The storage unit is made a battery cell to simplify the presentation. This is without loss of generality because a battery module or pack can be approximated by an upscaled cell model. The converter in each module is a synchronous DC/DC converter, which involves an inductor, a capacitor, and two power switches. It conducts controlled bidirectional power processing to charge or discharge the cell. The synchronous converters are used as they are sufficient for the power control need and also simple enough, though one may also use other types of DC/DC converters. These modules are connected in series or parallel configurations to fulfill the output voltage and power requirements. This BESS circuit structure is a derivative of what our previous study presents in \cite{2023-TTE-FA}.

As shown above, the circuit structure is characterized by the cell-level charging/discharging power control capability, making it possible to manage the power allocation for each individual cell. Harnessing this capability, one can enable a balanced utilization of the cells in terms of their SoC and temperature, while making them collectively meet a power demand. To achieve this, we will perform OPM design on the basis of the circuit structure from a control perspective.

Associated with the circuit structure, we propose a control architecture, as also illustrated in Fig.~\ref{FIG-1}. The architecture has two levels.

\begin{enumerate}[1)]
	\item The top level deals with OPM, which is the primary focus of our study. It solves an NMPC problem to determine the optimal PSR for each cell. The NMPC problem is formulated to minimize system-wide power losses, account for the operation conditions of individual cells, and adhere to constraints for safety, balancing, and power supply/demand match. 
	\item The lower-level controller computes the reference charging/discharging power for each cell. The reference computation is guided by the optimal PSR while also taking into account the instantaneous output power demand. This will make the cells contribute together to meeting the demand based on their own operating conditions, even when the demand is uncertain and vary from time to time. The lower-level controller further implements a PI controller to achieve reference tracking.
\end{enumerate}
\section{OPM Problem Formulation}
In this section, we elaborate on the OPM problem setup. We begin with showing an electrical and thermal model for the considered BESS. Next, we introduce the PSR and formulate an NMPC problem to perform PSR-based OPM.

\subsection{Electro-thermal Modeling}

\begin{figure}[!t]
    \centering
    \subfloat[\centering ]{{\includegraphics[width=4.9cm]{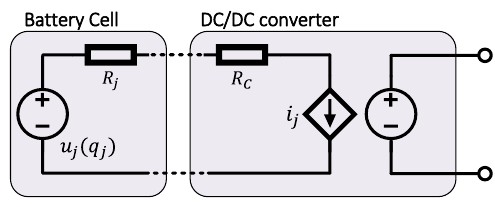} }}
    \subfloat[\centering ]{{\includegraphics[width=3.9cm]{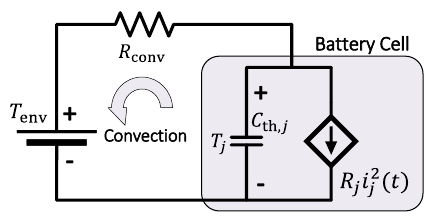} }}
    \caption{ (a) The electrical model of   module $j$; (b) the thermal model of cell $j$.} % The electro-thermal model, adapted from \cite{2024-TTE-FA}.
    \label{FIG-2}
\end{figure}

Fig.~\ref{FIG-2} (a) illustrates the electrical model of module $j$, which consists of cell $j$ and a DC/DC converter. We utilize the Rint model to describe cell $j$'s dynamics \cite{2011-Energies-HH}. This model includes an SoC-dependent open-circuit voltage (OCV) source and a series internal resistor, governed by
\begin{subequations}
	\begin{align}
		\dot{q}_j(t)&=-\frac{1}{Q_j}i_j(t),  \label{SoCDynamics_0}\\
		v_j(t)&=u_j(q_j(t))-R_j(q_j(t))i_j(t),
	\end{align}
\end{subequations}
where $q_j$, $Q_j$, $v_j$, $u_j$, $R_j$, and $i_j$ are the SoC, capacity, terminal voltage, OCV, internal resistor, and charging/discharging current, respectively. Note that $u_j$ and $R_j$ are dependent on $q_j$ for batteries. This dependency allows for a more accurate representation of the cell’s behavior. The cell’s internal charging/discharging power is
\begin{equation}
	P_{b_j}=u_j(q_j(t))i_j(t).
	\label{Internal_Power}
\end{equation}
The DC/DC converter is modeled as an ideal DC transformer connected with a series resistor $R_C$ to account for its power loss. The power loss of module $j$ can now be given by
\begin{equation}
	P_{L_j}=(R_j(q_j(t))+R_C)i_j^2(t).
	\label{PowerLoss_0}
\end{equation}

Proceeding forward, we introduce the PSR as follows:
\begin{equation}
	\mu_j=\frac{P_{b_j}}{P_{\textrm{out}}},
	\label{UtilFactor}
\end{equation}
where $\mu_j$ and $P_{\textrm{out}}$ denote the PSR and output power prediction, respectively. 
%{\color{red}In \eqref{UtilFactor}, we consider the absolute value of $P_{\textrm{out}}$ to accommodate both charging and discharging scenarios [Should we delete this? This statement is a bit unclear, and we only need to let $P_b$ and $P_{out}$ have the same sign in charging or discharging]}. 
The PSR indicates how much cell $j$ should contribute to fulfilling the output power demand. With the PSR definition in \eqref{UtilFactor}, we rewrite \eqref{SoCDynamics_0} and \eqref{PowerLoss_0} as follows:
\begin{subequations}
	\begin{align}
		\dot{q}_j(t)&=-\frac{\mu_j(t)P_{\textrm{out}}(t)}{Q_j\mathcal{Z}_j(q_j(t))}, \label{SoCDynamics}\\
		P_{L_j}&=\frac{(R_j(q_j(t))+R_C)\mu_j^2(t)P_{\textrm{out}}^2(t)}{\mathcal{Z}_j^2(q_j(t))}, \label{PowerLoss}
	\end{align}
	\label{ElectricalModel}
\end{subequations}
\hspace{-5pt}where $\mathcal{Z}_j(\cdot)$ encodes the SoC/OCV relationship, i.e., $u_j(t)=\mathcal{Z}_j(q_j(t))$.

To describe cell $j$'s temperature, we adopt a lumped thermal circuit model, as illustrated in Fig.~\ref{FIG-2} (b). This model focuses on capturing the effects of the heat generation induced by power losses and the convection with the ambient environment. It is governed by
\begin{equation}
	C_{\textrm{th}}\dot{T}_j(t) = R_j(q_j(t))i_j^2(t) - (T_j(t)-T_{\textrm{env}})/R_{\textrm{conv}},
	\label{ThermalModel-1}
\end{equation}
where $T_j$, $T_{\textrm{env}}$, $C_{\textrm{th}}$, and $R_{\textrm{conv}}$ are the cell's temperature, environmental temperature, heat capacity, and convective resistance, respectively. Following \eqref{Internal_Power} and \eqref{UtilFactor}, we can rewrite \eqref{ThermalModel-1} as
\begin{equation}
	C_{\textrm{th}}\dot{T}_j(t) = \frac{R_j(q_j(t))\mu_j^2(t)P_{\textrm{out}}^2(t)}{\mathcal{Z}_j^2(q_j(t))} - (T_j(t)-T_{\textrm{env}})/R_{\textrm{conv}}.
	\label{ThermalModel-2}
\end{equation}

From above, \eqref{ElectricalModel} and \eqref{ThermalModel-2} yield a concise yet representative electro-thermal model. Based on the model, we can now formulate an NMPC problem to achieve OPM. 
\subsection{NMPC-Based OPM Problem}
The OPM design aims to find the best PSR for all the cells throughout the operation of the BESS to minimize the system-level power losses, while satisfying different constraints pertaining to safety, balancing, and power supply/demand match. An NMPC-based problem formulation well suits the aim, exploiting NMPC's ability of running receding-horizon optimization under constraints.

To begin with, the overall power loss is given by
\begin{equation}
	L(t) = \sum_{j=1}^{n} \frac{(R_j(q_j(t))+R_C)\mu_j^2(t)P_{\textrm{out}}^2(t)}{\mathcal{Z}_j^2(q_j(t))}.
	\label{Loss}
\end{equation}
To ensure the safe operation of the BESS, we impose the subsequent constraints on the operating temperature, SoC, and charging/discharging currents of the cells:
\begin{subequations}
	\begin{align}
		T_j^{\min} &\leq T_j \leq T_j^{\max}, \label{Safety-1} \\
		q_j^{\min} &\leq q_j \leq q_j^{\max}, \label{Safety-2} \\
		i_{j}^{\min} &\leq i_j \leq i_{j}^{\max}, \label{Safety-3}
	\end{align}
	\label{SafetyConstraints}
\end{subequations}
\hspace{-5pt}where $T_{j}^{\textrm{min}/\textrm{max}}$, $q_{j}^{\textrm{min}/\textrm{max}}$, and $i_j^{\textrm{min}/\textrm{max}}$ are the lower/upper safety limits. Note that \eqref{Safety-3} can be expressed using the PSR as follows:
\begin{equation}
	\frac{\mathcal{Z}_j(q_j(t))}{|P_{\textrm{out}}(t)|} i_{j}^{\min} \leq \mu_j(t) \leq \frac{\mathcal{Z}_j(q_j(t))}{|P_{\textrm{out}}(t)|} i_{j}^{\max}.
	\label{Safety-4}
\end{equation}
We further introduce two constraints to promote the SoC and temperature balancing:
\begin{subequations}
	\begin{align}
		\left| q_j - q_{\textrm{avg}} \right| &\leq \Delta q, \\
		\left| T_j - T_{\textrm{avg}} \right| &\leq \Delta T,
	\end{align}
	\label{BalancingConstraints}
\end{subequations}
\hspace{-4pt}where $q_{\textrm{avg}}$ and $T_{\textrm{avg}}$ are the average SoC and temperature values across all cells, respectively, and $\Delta q$ and $\Delta T$ indicate the permissible deviations for individual cells from the average. We further enforce the following constraint to make the power supply meet with the demand:
\begin{equation}
	\sum_{j=1}^n \left(P_{b_j}(t) - (R_j(q_j(t))+R_C)i_{j}^2(t)\right) = P_{\textrm{out}}(t).
	\label{PowerBalanceConstraint-1}
\end{equation}
We again emphasize that $P_{\textrm{out}}(t)$ is the output power demand prediction. Using \eqref{PowerLoss}, we can express \eqref{PowerBalanceConstraint-1} in terms of PSR as below:
\begin{equation}
	\sum_{j=1}^n \mu_j(t) = 1 + \frac{L(t)}{|P_{\textrm{out}}(t)|}.
	\label{PowerBalanceConstraint-2}
\end{equation}
In \eqref{PowerBalanceConstraint-2}, we consider the absolute value of $P_{\textrm{out}}$ to accommodate both charging and discharging scenarios. The above equation reveals key properties of the PSR. Without power losses, the sum of PSR equals 1, and for identical cells, each will have a PSR of $1/n$. However,  power losses are practically inevitable, leading the sum of PSR to exceed 1 to account for the losses within the BESS.

Having laid out the power loss expression and the constraints, we are now in a good position to introduce the NMPC-based OPM problem. For concise presentation, we define $ \boldsymbol{\mu}(t) = \begin{bmatrix}\mu_1(t)&\dots&\mu_n(t)\end{bmatrix}^\top$. We then can express \eqref{Loss} as $L(t) = \boldsymbol{\mu}^\top(t) \boldsymbol{B}(t)\boldsymbol{\mu}(t)$, where $\boldsymbol{B}(t) \in \mathbb{R}^{n\times n}$ is a diagonal matrix with the diagonal elements being
\begin{equation}
	B_{jj} = \frac{(R_j(q_j(t)+R_C)P_{\textrm{out}}^2(t)}{\mathcal{Z}_j^2(q_j(t))}, \nonumber
\end{equation}
for $j=1,\dots,n$. The NMPC problem for OPM is formulated as follows:

\begin{equation}
	\begin{split}
		\min_{\mu_j ,j=1,...,n} &\quad \int_{k}^{k+H} L(t)dt, \\
		\textrm{s.t.} &\quad \eqref{SoCDynamics}, \eqref{ThermalModel-2}, \eqref{SafetyConstraints}\mbox{-}\eqref{BalancingConstraints}, \eqref{PowerBalanceConstraint-2},
	\end{split}
	\label{OptimNonConvex}
\end{equation}
where $H$ is the planning horizon length. The objective of \eqref{OptimNonConvex} is to optimize the PSR of the cells, whereby the optimal power is allocated to ensure minimized power losses while satisfying the specified constraints. Here, \eqref{OptimNonConvex} presents a nonlinear nonconvex optimization problem, the solution of which typically relies on numerical algorithms. Involving $3Hn-n$ optimization variables, this problem is computationally amenable only for small-scale BESS. It will require intensive computation while facing slow convergence, when a BESS comprises a large number of cells. We hence must find a way to surmount the computational complexity to facilitate the practical OPM implementation, which is the focus of what follows.
\subsection{Parameterized Control Policy}
In~\eqref{OptimNonConvex}, PSR is assigned to every cell as the control variable, and we propose a parametrized control policy, which determines a cell's PSR based on the BESS's operational condition. Note that, by design, the problem in \eqref{OptimNonConvex} attempts to find optimal PSR to minimize power losses while achieving SoC and temperature balancing. This nature implies that a cell's PSR essentially depends on its SoC, temperature, and internal resistance relative to the other cells. Specifically for discharging, the cell should have a higher PSR if its SoC is higher, its temperature is lower, and its internal resistance is less than other cells. The only difference in charging is that low SoC will imply higher PSR. With this rationale, we can parameterize the control policy as
\begin{equation}
	\mu_j(t) = \begin{bmatrix} \phi_{q, j}(t) & \phi_{\Taxi, j}(t) & \phi_{\Rap, j}(t) \end{bmatrix} \begin{bmatrix} \theta_1 \\ \theta_2 \\ \theta_3 \end{bmatrix},
	\label{Form-1}
\end{equation}
where $\theta_1$, $\theta_2$, and $\theta_3$ are the parameters with
\begin{equation}
	\theta_1 + \theta_2 + \theta_3 = 1. 
\end{equation}
The functions $\phi_{q, j}$, $\phi_{\Taxi, j}$, and $\phi_{\Rap, j}$ for cell $j$ are designed as follows: 
\begin{subequations}
	\begin{align}
		\phi_{q,j}(t) &=\begin{cases}\sum_{i=1}^{n}\left(\frac{q_j(t)}{q_i(t)}\right)^{-\xi_q} & \textrm{if} \; \; P_{\textrm{out}} \geq 0 \\ 										\sum_{i=1}^{n}\left(\frac{q_i(t)}{q_j(t)}\right)^{-\xi_q} & \textrm{if} \; \; P_{\textrm{out}} < 0\end{cases}, \label{q-Function}\\
		\phi_{\Taxi,j}(t) &= \sum_{i=1}^{n}\left(\frac{T_j(t)}{T_i(t)}\right)^{-\xi_\Taxi}, \label{T-Function} \\
		\phi_{\Rap,j}(t) &= \left(\sum_{i=1}^{n}\frac{R_j(q_j(t))+R_C}{R_i(q_i(t))+R_C}\right)^{-1}. \label{R-Function}
	\end{align}
	\label{Functions}
\end{subequations}
\hspace{-5pt}Here, the hyperparameters $\xi_q$ and $\xi_\Taxi$ weight the influence of the SoC and temperature on the PSR, relative to that of the internal resistance. The parameterized control policy in \eqref{Form-1} captures the key fact of the underlying solution to \eqref{OptimNonConvex}---cell $j$'s PSR depends on its SoC, temperature and internal resistance and also its operating condition compared to the other cells. The parameterization  reduces the $n$-dimensional control space based on $\mu_j$ for $j=1,2,\ldots,n$ to a three-dimensional parameter space defined by $\theta_1$, $\theta_2$, and $\theta_3$. This transformation greatly enhances search and computational efficiency. Further, the scalability in OPM design will improve considerably, because the dimensionality of the parameter space is decoupled from  the number of cells within the BESS.

We are now ready to reformulate \eqref{OptimNonConvex} according to the parameterized control policy defined in \eqref{Form-1}. For ease of notation, 
we denote \eqref{Form-1} as $\mu_j(t)=\boldsymbol{\phi}_j^\top(t)\boldsymbol{\theta}$,  where $\boldsymbol{\theta}=\begin{bmatrix}\theta_1&\theta_2&\theta_3\end{bmatrix}^\top$ and $\boldsymbol{\phi}_j(t)=\begin{bmatrix}\phi_{q,j}(t)&\phi_{\Taxi,j}(t)&\phi_{\Rap,j}(t)\end{bmatrix}^\top$. Consequently, we have $\boldsymbol{\mu}(t)=\boldsymbol{\Phi}(t)\boldsymbol{\theta}$ where $\boldsymbol{\Phi}(t) = \begin{bmatrix} \boldsymbol{\phi}_1(t) & \dots & \boldsymbol{\phi}_n(t)\end{bmatrix}^\top$. %Also note that $L(t) = \boldsymbol{\theta}^\top\boldsymbol{\Phi}^\top(t)\boldsymbol{B}(t)\boldsymbol{\Phi}(t)\boldsymbol{\theta}$.
Substituting \eqref{Form-1} into \eqref{OptimNonConvex} and discretizing \eqref{SoCDynamics} and \eqref{ThermalModel-2} using the forward Euler method, we derive the discretized NMPC-based OPM problem:
\begin{equation}
	\begin{aligned}
		&\min_{q_j, T_j, \boldsymbol{\theta}} \quad \sum_{t=k}^{k+H}\boldsymbol{\theta}^\top\boldsymbol{\Phi}^\top[t]\boldsymbol{B}[t]\boldsymbol{\Phi}[t]\boldsymbol{\theta} + \left( \bm \theta - \bar {\bm \theta} \right)^\top {\bm R} \left( \bm \theta - \bar {\bm \theta} \right), %  \left\| \bm \theta - \bar {\bm \theta} \right\|_{\bm R}^2,
\\
		&\textrm{SoC dynamics:} \quad \\
		&\ q_j[t+1] = q_j[t] - \frac{\boldsymbol{\phi}_j^\top[t]\boldsymbol{\theta}P_{\textrm{out}}[t]}{Q_j\mathcal{Z}_j(q_j([t])}\Delta t \\
		&\textrm{Temperature dynamics:} \quad \\
		&\ T_j[t+1]=T_j[t] + \frac{\Delta t}{C_{\textrm{th}}} \bigg[\frac{R_j(q_j(t))(\boldsymbol{\phi}_j^\top[t]\boldsymbol{\theta})^2P_{\textrm{out}}^2[t]}{\mathcal{Z}_j^2(q_j[t])} \\
		&\ \qquad \qquad - (T_j[t]-T_{\textrm{env}})/R_{\textrm{conv}} \bigg],\\
		&\textrm{Safety constraints:} \quad \\
		&\ T_j^{\textrm{min}}\leq T_j\leq T_j^{\textrm{max}},\\
		&\ q_j^{\min} \leq q_j \leq q_j^{\max}, \\
		&\ \frac{\mathcal{Z}_j(q_j[t])}{|P_{\textrm{out}}[t]|} i_{j}^{\min} \leq \boldsymbol{\phi}_j^\top[t]\boldsymbol{\theta} \leq \frac{\mathcal{Z}_j(q_j[t])}{|P_{\textrm{out}}[t]|} i_{j}^{\max},\\
		&\textrm{Balancing constraints:} \quad \\
		&\ \left| q_j - q_{\textrm{avg}} \right| \leq \Delta q,\\
		&\ \left| T_j - T_{\textrm{avg}} \right| \leq \Delta T,\\
		&\textrm{Parameter constraint:} \quad \\
		&\ {\bm 1}^\top\boldsymbol{\theta} = 1,\\
		&\textrm{Power supply-demand balance:} \quad \\
		&\ {\bf 1}^\top {\bm \Phi[t]} \bm \theta = 1 + \frac{\boldsymbol{\theta}^\top\boldsymbol{\Phi}^\top[t]\boldsymbol{B}[t]\boldsymbol{\Phi}[t]\boldsymbol{\theta}}{|P_{\textrm{out}}[t]|}, \\
%&\mbox{\color{red} Is the definition of $\Phi$ right? We should also make $\phi_j$ bold.}\\
%& {\color{red} \sum_{j=1}^n \boldsymbol{\phi}_j^\top[t]\boldsymbol{\theta} = {\bf 1}^\top {\bm \Phi[t]} \bm \theta ?} 
	\end{aligned} 
		\label{NMPC}
\end{equation}
where $t$ is the discrete time index and $\Delta t$ is the sampling time, and ${\bm 1}$ is a column vector of ones. In above, $\boldsymbol{\theta}^\top\boldsymbol{\Phi}^\top[t]\boldsymbol{B}[t]\boldsymbol{\Phi}[t]\boldsymbol{\theta}$ corresponds equivalently to $L[t]$, aligning the formulation in \eqref{NMPC} with the objective of minimizing the power losses. The cost function also includes a regularization term, $\left( \bm \theta - \bar {\bm \theta} \right)^\top {\bm R} \left( \bm \theta - \bar {\bm \theta} \right)$,  to stabilize the optimization process, as suggested in \cite{Arxiv-JA-2022, SIAM-CN-2020}, where $\bar {\bm \theta}$ represents the nominal value of $\boldsymbol{\theta}$ and ${\bm R}$ is a weighting factor. This regularization term will play a role in stabilizing the estimation process later. The other parts of \eqref{NMPC} also find their exact counterparts in \eqref{OptimNonConvex}. This new formulation in~\eqref{NMPC} preserves fidelity to the original OPM problem, while mandating the search for the optimal $\bm \theta \in \mathbb{R}^3$ rather than $\bm \mu\in \mathbb{R}^n$ as in~\eqref{OptimNonConvex} with $3 \ll n$. However, it still involve $2nH+3$ optimization variables due to the involvement of $q_j$ and $T_j$ for $j=1,\ldots,n$. This makes~\eqref{NMPC} still resistant to numerical gradient-based optimization for a large-scale BESS. To address this issue, we take a Bayesian inference approach as will be shown next.

\section{Bayesian Inferential OPM}
In this section,  we   recast  the OPM problem in \eqref{NMPC}  as a Bayesian inference problem and then leverage the computationally fast EnKI approach to perform the inference. 

\subsection{Bayesian Inference for the OPM Problem}

To begin with, we rewrite~\eqref{NMPC} compactly as below while moving the discrete-time index to the subscript for notational simplicity: 
\begin{subequations}\label{compact-OPM-problem}
\begin{align}
    \min_{\bm x_{k:k+H}, \bm \theta} \quad & \sum_{t=k}^{k+H} \bm \theta^\top \bm H(\bm x_t) \bm \theta + \left( \bm \theta - \bar {\bm \theta} \right)^\top {\bm R} \left( \bm \theta - \bar {\bm \theta} \right),  \\ \label{compact-state-dynamics}
    \mathrm{s.t.} \quad & \bm x_{t+1}  = \bm f \left( \bm x_t, \bm \theta \right),\\ \label{compact-inequality}
    & \bm g \left( \bm x_t, \bm \theta \right)   \leq \bm 0,\\  \label{compact-equality}
    & \bm e \left( \bm x_t, \bm \theta \right)  = \bm 0,
\end{align}
\end{subequations} 
where 
\begin{align*}
\boldsymbol{x}_t &= \left[\begin{matrix}q_{1}[t]&...&q_{n}[t] & T_{1}[t]&...&T_{n}[t]\end{matrix}\right]^\top,\\
\bm H(\bm x_t) &= \boldsymbol{\Phi}^\top_t\boldsymbol{B}_t\boldsymbol{\Phi}_t.
\end{align*}
In correspondence to~\eqref{NMPC},~\eqref{compact-state-dynamics}-\eqref{compact-equality} represent the cell dynamics, inequality-based safety and balancing constraints, and equality-based parameter and power supply-demand balance constraints, respectively. Given~\eqref{compact-OPM-problem}, the search for optimal $\bm \theta$ inspires an interesting question: {\em can we infer or estimate $\bm \theta$ if the system is observed to behave optimally as defined in~\eqref{compact-OPM-problem} within the horizon $t\in [k, k+H]$?} Answering this question will lead to an inference-driven solution to~\eqref{compact-OPM-problem}, which is  our focus next. 

Let us introduce an auxiliary virtual dynamic system as follows: 
\begin{equation} 
\left\{
		\begin{aligned}
		\boldsymbol{x}_{t+1} &= \boldsymbol{f}(\boldsymbol{x}_t,\boldsymbol{\theta}), \\
		\bm y_t &= \boldsymbol{h}(\boldsymbol{x}_t,
\boldsymbol{\theta}) + \bm v_t, 
		\end{aligned}
\right.
	\label{VirtualDynamics}
\end{equation}
for $t=k,...,k+H$, where
\begin{align} \label{h_Function}
    \bm h \left( \bm x_t, \bm \theta \right) =  
    \begin{bmatrix} \sqrt{\bm H(\bm x_t)} \bm \theta \cr \bm \psi_{\bm g} \left( \bm g \left( \bm x_t, \bm \theta \right) \right) \cr
    \bm \psi_{\bm e} \left( \bm e \left( \bm x_t, \bm \theta \right) \right)
    \end{bmatrix},
\end{align} 
where $\sqrt{(\cdot)}^\top \sqrt{(\cdot)} = (\cdot)$, $\bm \psi_{\bm g}(\cdot)$ and $\bm \psi_{\bm e}(\cdot)$ are barrier functions that output $\bm 0$ when the constraint is met and $\bm \infty$ otherwise, and $\bm v_t$ is an additive noise. This auxiliary system replicates the dynamics in~\eqref{compact-OPM-problem} and adds an observation equation. The observation variable $\bm y_t$ is designed to measure the level of optimality of the system's behavior and the level of constraint satisfaction, through the construction of $\bm h(\cdot)$. Comparing~\eqref{VirtualDynamics} against~\eqref{compact-OPM-problem}, we can find that $\bm y_t= \bm 0$ if the system behaves optimally. This suggests the potential of inferring  optimal $\bm \theta$ based on the virtual observations $\bm y_t= \bm 0$ for $t=k, \ldots, k+H$. To formalize the notion, we consider   the maximum a posteriori (MAP) estimation:
\begin{align}\label{MAP}
	\boldsymbol{\hat{\theta}}^*= \arg \max_{\boldsymbol{\theta}} \, \, \log \p(\bm x_{k:k+H}, \boldsymbol{\theta} | \boldsymbol{y}_{k:k+H}=\boldsymbol{0}).	
\end{align}
In sequel, we will drop $\bm 0$ and use $\p(\bm x_{k:k+H}, \boldsymbol{\theta} | \boldsymbol{y}_{k:k+H})$   to ease  notation. 

Next, we demonstrate how the MAP problem in \eqref{MAP} aligns with the OPM problem in~\eqref{compact-OPM-problem}.  
To this end, we set up a moving horizon estimation (MHE) problem for \eqref{VirtualDynamics} to estimate $\boldsymbol{x}_{k:k+H}$ and $\boldsymbol{\theta}$:
\begin{subequations}
\label{MHE}
\begin{align}
    \min_{\bm x_{k:k+H}, \bm \theta} \quad & \sum_{t=k}^{k+H}  \bm v_t^\top {\bm Q} \bm v_t + \left( \bm \theta - \bar {\bm \theta} \right)^\top {\bm R} \left( \bm \theta - \bar {\bm \theta} \right),  \\
    \mathrm{s.t.} \quad & \bm x_{t+1}  = \bm f \left( \bm x_t, \bm \theta \right), \label{DynamicMHE} \\
     &\bm y_t = \bm h 
    \left( \bm x_t, \bm \theta \right) +\bm v_t. 
\end{align}
\end{subequations} 
This MHE problem is nearly identical to~\eqref{NMPC}. The only difference lies in the treatment of the constraints. Here, the constraints are enforced by adding to the cost function a penalty term based on the barrier functions $\bm \psi_{\bm g}(\cdot)$ and $\bm \psi_{\bm e}(\cdot)$. The following theorem shows how~\eqref{MAP} relates to~\eqref{MHE}.

\begin{theorem}\label{MAP-MHE-Equivalence}
Assume that $ \boldsymbol{\theta}  \sim \mathcal{N}(\boldsymbol{\bar{\theta}}, \bm R^{-1})$, and that $ \bm v_t \sim\mathcal{N}(\bm 0, \bm Q^{-1})$ is a  white Gaussian noise.
Then, the problems in \eqref{MHE} and \eqref{MAP} share the same optima.
\end{theorem}

\begin{proof}
Using   Bayes' rule and the Markovian property of \eqref{VirtualDynamics}, we have
\begin{align}
	& \p( { \bm x_{k:k+H}, \boldsymbol{\theta} } | \boldsymbol{y}_{k:k+H})\nonumber\\
	& \propto \p(\bm y_{k:k+H} | { \bm x_{k:k+H}, \bm \theta } )\p(\bm x_{k:k+H} | \bm \theta)p(\bm \theta)\nonumber\\
	& = \prod_{t=k}^{k+H} \p(\bm y_t | \bm x_t, \boldsymbol{\theta}) \cdot \prod_{t=k}^{k+H-1} \p(\bm x_{t+1} | \bm x_{t}, \boldsymbol{\theta}) \cdot \p(\bm x_k | \bm \theta)\cdot p(\bm \theta). \nonumber
\end{align}
Since \eqref{DynamicMHE} is deterministic, we have 
\begin{align*}
\p(\bm x_{t+1} | \bm x_t, \boldsymbol{\theta})  &= \mathds{1} \left[ \bm x_{t+1} - \bm f \left(\bm x_t, \bm \theta \right) \right],  
\end{align*}
where $\mathds{1} [ \cdot ]$ is the indicator function. Also,  $\bm x_k$ is known and assumed to take $\bm x_k^o$, then $\p( x_k  | \bm \theta) = \mathds{1} \left[ \bm x_k -  \bm x_k^o\right]$. It then follows that
\begin{align}
	\log \p( { \bm x_{k:k+H}, \boldsymbol{\theta} } | \boldsymbol{y}_{k:k+H}) \propto \sum_{t=k}^{k+H} \log \p(\bm y_t | { \bm x_t,  \boldsymbol{\theta} } ) + \log p(\boldsymbol{\theta}).
	\nonumber
\end{align}
As $\p(\bm y_t | \bm x_t, \boldsymbol{\theta})\sim\mathcal{N}\left(\bm h \left( \bm x_t, \bm \theta \right), \bm Q^{-1} \right)$ and $ \bm \theta  \sim \mathcal{N}(\boldsymbol{\bar{\theta}}, \bm R^{-1})$, we get
\begin{align*}
	\log \p(\bm y_t | \bm x_t, \boldsymbol{\theta}) &\propto     - \left( \bm y_t - \bm h \left(\bm x_t, \bm \theta\right) \right)^\top  \bm Q  \left( \bm y_t - \bm h \left(\bm x_t, \bm \theta\right) \right) \\ 
&=  - \bm v_t^\top  \bm Q \bm v_t,\\
  \log p(\bm \theta) &\propto - \left( \bm \theta - \bar {\bm \theta} \right)^\top \bm R \left( \bm \theta - \bar {\bm \theta} \right). 
\end{align*} 
Putting together the above, we see that the reward function in \eqref{MAP} is the opposite of the cost function in \eqref{MHE}. The theorem is thus proven.
\end{proof}
Theorem~\ref{MAP-MHE-Equivalence}  reveals that, in a Gaussian setting, the   MAP problem in~\eqref{MAP}  is equivalent to~\eqref{MHE}, which recalls the duality between control and estimation~\cite{2008-CDC-TE}. As such, by solving~\eqref{MAP}, we can address the original OPM problem. However, no analytical solution is available to~\eqref{MAP}, and we will leverage the EnKI method to approximately solve \eqref{MAP}.
 
\subsection{OPM via EnKI}

The EnKI method is a Monte Carlo computational framework for Bayesian parameter inference, which grows out of the well-known ensemble Kalman filtering technique. At the core, the EnKI computes ensemble-based approximation of the  posterior distribution of parameters, and iteratively evolves the ensemble using the Kalman update to improve the approximation accuracy~\cite{IP-NK-2019,StatLetters-SD-2022}. Computation-wise, it is fast, derivative-free, and empirically stable. These features make the EnKI particularly suitable for handling parameter inference for high-dimensional nonlinear systems. 

%Rooted in ensemble Kalman filtering, the ensemble Kalman inversion (EnKI) method provides a sampling-based approximation of the target posterior distribution \cite{IP-NK-2019}. It is particularly effective for estimating unknown parameters in complex systems based on measurement data. This method offers several advantages: it is derivative-free, computationally efficient, and leverages sampling to accelerate convergence. These features make it especially suitable for large-scale estimation problems, where the high dimensionality and nonlinearity of the system would otherwise render traditional techniques impractical. 

Consider  $\p(\bm x_{k:k+H}, \boldsymbol{\theta} | \boldsymbol{y}_{k:k+H})$. Since our focus is on estimating $\boldsymbol{\theta}$, we can marginalize out $\bm x_{k:k+H}$: 
\begin{align}\label{Marginal-Posterior}\nonumber 
    &\p( \bm \theta | \bm y_{k:k+H}) = \int \p( { \bm x_{k:k+H}, \bm \theta} | \bm y_{k:k+H}) d \bm x_{k:k+H} \nonumber \\
    &\propto \underbrace{\int \p(\bm y_{k:k+H} | {\bm x_{k:k+H},   \bm \theta } ) \p(  \bm x_{k:k+H} | \bm \theta) d \bm x_{k:k+H}}_{ \p(\bm y_{k:k+H} | \bm \theta) } \cdot p\left( \bm \theta \right).  
\end{align} 
Here, to enable   iterative approximation of $\p( \bm \theta | \bm y_{k:k+H})$,  we consider building tempered posteriors and  rewrite~\eqref{Marginal-Posterior} as  
\begin{align} 
    \p( \bm \theta | \bm y_{k:k+H}) \propto \prod_{\ell=1}^N    \left[ \p( \bm y_{k:k+H} |  \bm \theta ) \right] ^{\lambda_\ell} \cdot  p\left( \bm \theta \right),
    \label{TemperedPosterior}
\end{align}
where  $\ell$ is the iteration index, $[\cdot]^{\lambda_\ell}$ denotes the power of $\lambda_\ell$, and $  \lambda_{\ell} \in (0, 1)$ is a hyperparameter used to temper the likelihood, with $\sum_{\ell=1}^N \lambda_{\ell}=1$. As will be seen later, $\lambda_{\ell}$ will play a step-size role   in the ensemble  update procedure \cite{MWR-PH-2016}.  We define $\bm Y = \bm y_{k:k+H}$ for   notational simplicity and following \eqref{TemperedPosterior}, obtain
\begin{align} \label{Tempered-Posterior-Iteration}
    p_{\ell+1}( \bm \theta \, | \, \bm Y ) \propto \left[  \p(  \bm Y | \bm \theta) \right]^{\lambda_\ell} \cdot  p_\ell \left( \bm \theta  \, | \, \bm Y \right),
\end{align}
where $p_\ell \left( \bm \theta \, | \, \bm  Y \right)$ is the tempered posterior at iteration  $\ell$, and $p_1 \left( \bm \theta \, | \, \bm Y\right) = p(\bm \theta) \sim \mathcal{N}\left(\boldsymbol{\bar{\theta}}, \bm R^{-1} \right)$.

Suppose that $p_\ell \left( \bm \theta \, | \, \bm Y \right)$ and $p_\ell \left( \bm \theta \, | \, \bm Y \right)$ are  weakly approximated  by  two ensembles of particles or samples, $\bm \theta_\ell^i$ and $\bm Y_\ell^i$ for $i=1,\ldots, M$, respectively. Note that $\bm Y_\ell^i$ results from $\bm \theta_\ell^i$, as will be explained later.   Guided by~\eqref{Tempered-Posterior-Iteration}, the EnKI will perform a Kalman update of the ensemble for $p_{\ell+1} \left( \bm \theta \, | \, \bm Y \right)$ as follows:
\begin{align}\label{sample-update}
\boldsymbol{\theta}_{\ell+1}^i= \boldsymbol{\theta}_{\ell}^i +\bm{\Sigma}^{\boldsymbol{\theta} \bm{y}}_\ell \left( \bm{\Sigma}^{\bm{y}}_\ell +  \lambda_{\ell}^{-1}\bar{\bm{Q}} \right)^{-1} \left(\bm Y -\bm{Y}_{\ell}^i \right),
\end{align}
for $\ell=1,\ldots, N$ and $i=1, \ldots, M$, where
\begin{subequations}\label{ensemble-based-computation}
	\begin{align}
		\hat{\bm \theta}_\ell &= \frac{1}{M}\sum_{i=1}^M \boldsymbol{\theta}_{\ell}^i, \quad \widehat{ \bm{ Y} }_{\ell} = \frac{1}{M}\sum_{i=1}^M \bm{Y}_{\ell}^i, \label{Sample-Mean}\\
		\bm{\Sigma}^{\bm \theta}_\ell &= \frac{1}{M-1}\sum_{i=1}^M \left( \boldsymbol{\theta}_\ell^i -   \hat{\bm \theta}_{\ell}\right)\left(  {\bm \theta}_\ell^i -   \hat{\bm  \theta}_{\ell}\right)^{\top}, \label{Sample-Cov}\\
		\bm{\Sigma}^{\bm{y}}_\ell &= \frac{1}{M-1}\sum_{i=1}^M \left( {\bm Y}_\ell^i -  \widehat{\bm Y}_{\ell} \right) \left( {\bm Y}_\ell^i -  \widehat{\bm Y}_{\ell} \right)^{\top}, \\
		\bm{\Sigma}^{\boldsymbol{\theta}\bm{y}}_\ell  &= \frac{1}{M-1}\sum_{i=1}^M \left( {\bm \theta}_\ell^i -  \hat{\bm \theta}_{\ell}\right)\left( {\bm Y}_{\ell}^i -  \widehat{\bm Y}_{\ell} \right)^{\top},\\ \label{Sample-CrossCov}
\bar{\boldsymbol{Q}} &= \mathrm{diag}(\bm Q^{-1}, \ldots, \bm Q^{-1}).
	\end{align}
\end{subequations}
In above, $\hat{\bm \theta}_\ell$ and $\hat{\bm Y}_\ell$ are the sample means, and $\bm{\Sigma}^{\bm \theta}_\ell$, $\bm{\Sigma}^{\bm Y}_\ell$ and $\bm{\Sigma}^{\bm \theta \bm Y}_\ell$ are the sample covariances based on the ensembles. Generally, $\lambda_{\ell}$ takes a small value in the early iterations, and increases as the iteration moves forward. One can also use a bisection search algorithm to identify the best $\lambda_{\ell}$ for each iteration \cite{StatLetters-SD-2022}. Upon the end of the iterative procedure when $\ell = N$, we can compute and output $\hat {\bm \theta}$ and $\bm{\Sigma}^{\bm \theta} $ as
\begin{subequations}\label{theta-estimate-output}
\begin{align}\label{Samples-Posterior}
\hat{\bm \theta}^*  &= \frac{1}{M}\sum_{i=1}^M \boldsymbol{\theta}_{N+1}^i,\\
\bm{\Sigma}^{\bm \theta} &= \frac{1}{M-1}\sum_{i=1}^M \left( \boldsymbol{\theta}_{N+1}^i -   \hat{\bm \theta}^* \right)\left(  {\bm \theta}_{N+1}^i -   \hat{\bm  \theta}^* \right)^{\top}. 
\end{align}
\end{subequations}
Here, $\bm{\Sigma}^{\bm \theta}$ quantifies the uncertainty in the estimation, and its computation can be omitted to enhance efficiency if a user has no interest in assessing the uncertainty.

To complete the computation in~\eqref{sample-update}-\eqref{ensemble-based-computation}, we must  compute $\bm Y_\ell^i$ for $i=1,\ldots,M$ at every iteration. One can regard $\bm Y_\ell^i$ as samples drawn from $\p(\bm y_{k:k+H} | \bm \theta)$. For the sake of clarity, we recall~\eqref{Marginal-Posterior} and restate
\begin{align} \label{y-theta-likelihood} \nonumber
\p(\bm y_{k:k+H} | \bm \theta)  & = \int \p(\bm y_{k:k+H} | { \bm x_{k:k+H},   \bm \theta} ) \\ 
& \quad\quad \quad \cdot \p(  \bm x_{k:k+H} | \bm \theta) d \bm x_{k:k+H}.
\end{align} 
Since 
\begin{align*}
\p(  \bm x_{k:k+H} | \bm \theta) &\propto \prod_{t=k}^{k+H-1} \p( \bm x_{t+1} | {\bm x_t, \bm \theta} ) \cdot \p(\bm x_k | \bm \theta)\\ 
& = \prod_{t=k}^{k+H-1}  \mathds{1} \left[ \bm x_{t+1} - \bm f \left(\bm x_t, \bm \theta \right) \right] \cdot \mathds{1} \left[ \bm x_k -  \bm x_k^o\right],
\end{align*}
we can compute  the ensemble for $\p(   \bm x_{k:k+H} |   \bm \theta  ) $ by
\begin{align}\label{Traj-Gen}
\bm {x}^i_{t+1,\ell} = \boldsymbol{f}\left(\boldsymbol{x}^i_{t,\ell}, \boldsymbol{\theta}_{\ell}^i \right), \quad i = 1, \ldots, M,
\end{align}
for $t=k,\ldots,k+H-1$. 
Then, by~\eqref{y-theta-likelihood}, the ensemble for $\p(\bm y_{k:k+H} | \bm \theta)$ can be computed by
\begin{align}\label{Meas-Gen}
\bm y_{t,\ell}^i  = \boldsymbol{h}(\boldsymbol{x}_{t,\ell}^i,\boldsymbol{\theta}_\ell^i) + \bm v_{t,\ell}^i,
\end{align}
for $t=k,\ldots,k+H$, where $\bm v_{t,\ell}^i \sim \mathcal{N}(\bm 0, \bm R)$. 

The above shows the computational procedure of the EnKI method to infer optimal $\bm \theta$ for solving the OPM problem. We name the obtained algorithm as \texttt{EnKI-OPM}, summarizing it in Algorithm~\ref{EnKI-OPM}. This algorithm harnesses the power of sampling and Kalman ensemble update to quickly search the parameter space. It thus offers much faster computation than gradient-based optimization, which makes it competent for handling the OPM problem for large-scale BESS. Our simulation results in Section~\ref{Simulation-Results} will further validate this. 

\begin{algorithm}[t]
\fontsize{9.2}{10}
  \caption{\texttt{EnKI-OPM} Ensemble Kalman inversion-based optimal power management} \label{NMPC-EnKS}
  \begin{algorithmic}[1]
\State Set up the OPM problem as in~\eqref{NMPC} 

\State Convert the problem to a Bayesian parameter inference problem as in~\eqref{MAP}
 
\State Initialize by setting  $\ell = 0$, $\hat{\bm \theta}_0 = \bar {\bm \theta}$, and $\bm{\Sigma}^{\boldsymbol{\theta}}_0 = \bm R^{-1}$
\State Draw samples $\boldsymbol{\theta}_0^i \sim  \mathcal{N}(\hat {\bm \theta}_0, \bm{\Sigma}^{\boldsymbol{\theta}}_0)$, for $i=1, \ldots, M$
%\While{$\Delta \boldsymbol{\theta}^\ell \geq \epsilon$} 
\For{$\ell = 1,\ldots, N$}
\For{$t = k,\ldots, k+H$}
\State Generate $\bm{x}_{t, \ell}^i$ using~\eqref{Traj-Gen}
\State Generate $\bm y_{t,\ell}^i$ as in~\eqref{Meas-Gen}
\EndFor
\State Construct $\bm Y_\ell^i =\begin{bmatrix}\bm y_{k,\ell}^i & \dots & \bm y_{k+H,\ell}^i \end{bmatrix}^\top$

\State Specify $\lambda_{\ell}$ via bisection search

\State Update $\boldsymbol{\theta}_{\ell+1}^i$ via~\eqref{sample-update}-\eqref{ensemble-based-computation}

%\State Compute $\boldsymbol{\bar \theta}^{\ell+1}$ using~\eqref{Samples-Posterior}

%\State Compute $\boldsymbol{\bar \theta}^{\ell}$ and $\bm{\bar y}^{\ell}$ via~\eqref{Sample-Mean}
%\State Compute $\bm{\Sigma}^{\boldsymbol{\theta}, \ell}$, $\bm{\Sigma}^{\bm{y}, \ell}$, and $\bm{\Sigma}^{\boldsymbol{\theta} \bm{y}, \ell}$ via~\eqref{Sample-Cov}-\eqref{Sample-CrossCov}
%\State Specify $\lambda^{\ell}$ via bisection search
%\State Update $\boldsymbol{\theta}^{\ell+1, i}$ via~\eqref{sample-update}
%\State Compute $\boldsymbol{\bar \theta}^{\ell+1}$ using~\eqref{Samples-Posterior}
%\State Compute criteria $\Delta \boldsymbol{\theta}^{\ell+1} \leftarrow \lVert\boldsymbol{\bar \theta}^{\ell+1} - \boldsymbol{\bar \theta}^{\ell}\rVert$
\State $\ell \leftarrow \ell + 1$
\EndFor

\State Compute $\hat{\bm \theta}^*$ via~\eqref{theta-estimate-output} 

\State Compute  the PSR $\bm \mu_k$ using~\eqref{Form-1}
\State Apply $\bm \mu_k$ to BESS
\end{algorithmic}
\label{EnKI-OPM}
\end{algorithm}

\section{Low-Level Controller}
The proposed OPM approach has computed the optimal PSR for the cells. As ratios, the PSR can be used to split the total charging/discharging power among the cells, even when the actual output power demand is uncertain and deviates from the prediction. To implement this approach, we must develop a lower-level controller that utilizes the PSR to determine the exact power quotas for the cells and then extract the power from each cell. Fig.~\ref{Low-Level} shows the lower-level controller and its interface with the higher-level OPM approach and the BESS.

The lower-level controller aims to regulate the terminal voltage of the BESS---once the voltage regulation is achieved, the control loop will drive the BESS towards power supply-demand balance. In other words, we regulate the BESS terminal voltage to estimate and overcome the uncertainties in the output power demand. Specifically, we employ the following proportional-integral (PI) controller:
\begin{equation}
	\tilde{P}_{\textrm{out}} = K_P(V^*_{\textrm{out}} - V_{\textrm{out}}) + K_I\int (V^*_{\textrm{out}} - V_{\textrm{out}})dt, 
\end{equation}
where $\tilde{P}_{\textrm{out}}$ quantifies the disparity between the predicted and actual output power demands; $V_{\textrm{out}}$ and $V^*_{\textrm{out}}$ denote the actual and desired BESS terminal voltage, respectively, and $K_P$ and $K_I$ are the coefficients for the proportional and integral terms, respectively. It is evident that, if the output power prediction ${P}_{\textrm{out}}$ matches the actual demand perfectly, then $V_{\textrm{out}} = V^*_{\textrm{out}}$ and $\tilde{P}_{\textrm{out}} = 0$. When there is a mismatch, the effective power demand imposed on the BESS is $P_{\textrm{out}} + \tilde{P}_{\textrm{out}}$.

Proceeding further, we can distribute $P_{\textrm{out}}+\tilde{P}_{\textrm{out}}$ among the cells using the PSR. Note that the power allocation must satisfy the individual cells' power limits to ensure safety, particularly when $\tilde{P}_{\textrm{out}}$ is large. We thus propose to split the power as follows:
\begin{equation}
	P_{b_j}^* = \max\left\{\min\{\mu_j^*(P_{\textrm{out}} + \tilde{P}_{\textrm{out}}), P_{b_j}^{\textrm{max}}\},P_{b_j}^{\textrm{min}}\right\},
\end{equation}
where $P_{b_j}^{\textrm{max}/\textrm{min}}=u_ji_j^{\textrm{max}/\textrm{min}}$ are cell $j$'s upper/lower power limits. Then, $P_{b_j}^*$ for $j=1, 2, \ldots, n$ are fed to the DC/DC converter controllers, which generate pulse-width modulation (PWM) signals to extract the exact power from each cell.

The above lower-level controller is designed to handle fast power distribution and regulation for the cells within the BESS. It integrates seamlessly with the proposed OPM approach to create a hierarchical control framework for system-to-cell power management. The two layers of this hierarchy can operate on different time scales: while the lower-level controller runs in real time, the OPM approach can run at slower rates. This is possible because the cells' slow dynamics and the gradual shifts in output power demand allow the PSR to be updated over longer intervals, thus reducing computational load.
\begin{figure*}[t]
	\centering
	\includegraphics[trim={0.2cm 0cm 0cm 0cm},clip,width=12cm]{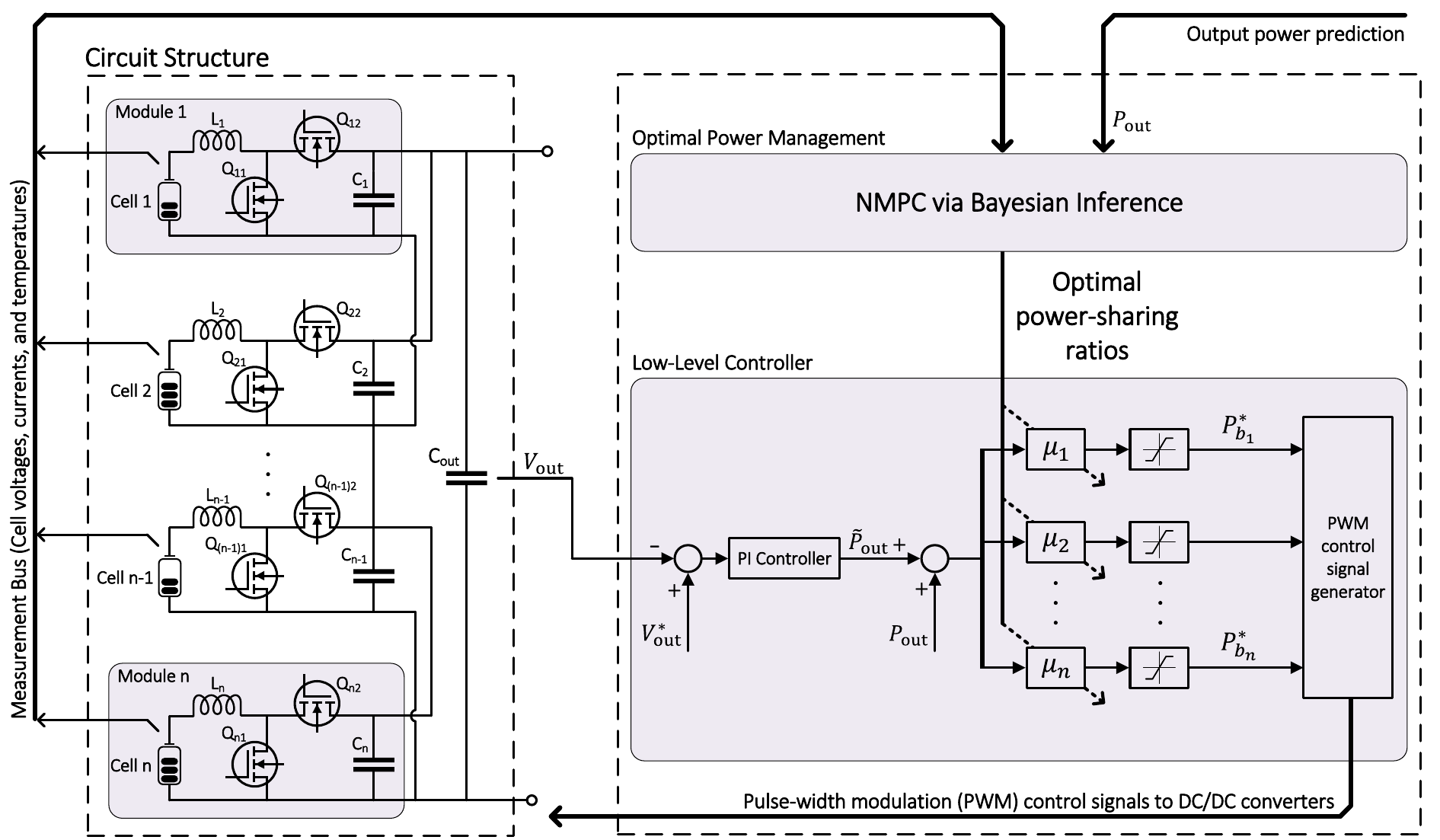}
      \caption{The interactions of the proposed optimal power management approach, low-level controller, and BESS.}
      \label{Low-Level}
\end{figure*}
\section{Simulation Results}\label{Simulation-Results}
\begin{table}[!t]
	\renewcommand{\arraystretch}{1.3}
	\caption{Specifications of the considered large-scale BESS}
	\centering
	\label{Table_1}
	%\centering
	\resizebox{\columnwidth}{!}{
		\begin{tabular}{l l l}
			\hline\hline \\[-3mm]
			\multicolumn{1}{c}{Symbol} & \multicolumn{1}{c}{Parameter} & \multicolumn{1}{c}{Value [Unit]}  \\[1.6ex] \hline
			$ n $ & Number of battery cells & 200 \\
			$ v $  & Cell nominal voltage & 3.6     [V] \\
			$ Q $ & Cell nominal capacity & 2.5     [Ah] \\ 
			$ R $ & Cell internal resistance & 31.3     [m$\Omega$] \\
			$ [q^{\textrm{min}},q^{\textrm{max}}] $ & Cell SoC limits  & [0.05,0.95] \\ 
			$ [i^{\textrm{min}},i^{\textrm{max}}] $ & Cell current limits & [-5,5]     [A] \\ 
			$ C_{\textrm{th}} $ & Specific thermal capacitance & 40.23    [J/K] \\ 
			$ R_{\textrm{conv}} $ & Convection thermal resistance & 41.05     [K/W] \\ 
			$ T_{\textrm{env}} $ & Environment temperature & 298     [K] \\ 
			$ \Delta q $ & SoC balancing threshold & 0.5\% \\
			$ \Delta T $ & Temperature balancing threshold & 1    [K] \\
			$ \Delta t $ & Sampling time & 1     [s]\\
			$ H $ & Horizon length & 10     [s]\\
			\hline\hline
		\end{tabular}
	}
\end{table}

This section presents the simulation results for our proposed OPM approach. We consider a large-scale BESS comprising 200 modules, the specifications of which are summarized in Table~\ref{Table_1}. The battery cells are based on Samsung INR18650-25R cells whose parameters are identified in \cite{2020-JES-NT}. These cells' SoC/OCV relationship and internal resistance are described by
\begin{align}
	u(q) &= 3.3 + 2.61q - 9.36q^2 + 19.7q^3 - 19.0q^4 + 6.9q^5, \nonumber \\
	R(q) &= 0.0313 + 0.0678e^{-13.2q}. \nonumber
\end{align}
The output power demand is approximately 2 kW, alternating between charging and discharging every 20 minutes. It is subject to a slowly evolving noise following a uniform distribution  
$\mathcal{U}(-200, 200)$ W. This noise term is intended to simulate the deviation between the predicted and actual power demands, as shown in Fig.~\ref{FIG_SIM_1}. The simulations cover one hour of BESS operation, with a sampling time of one second. The OPM's horizon length is  $H=10$ seconds.

\begin{figure}[!t]\centering
	\includegraphics[trim={0cm 0cm 0cm 0cm},clip,width=\linewidth]{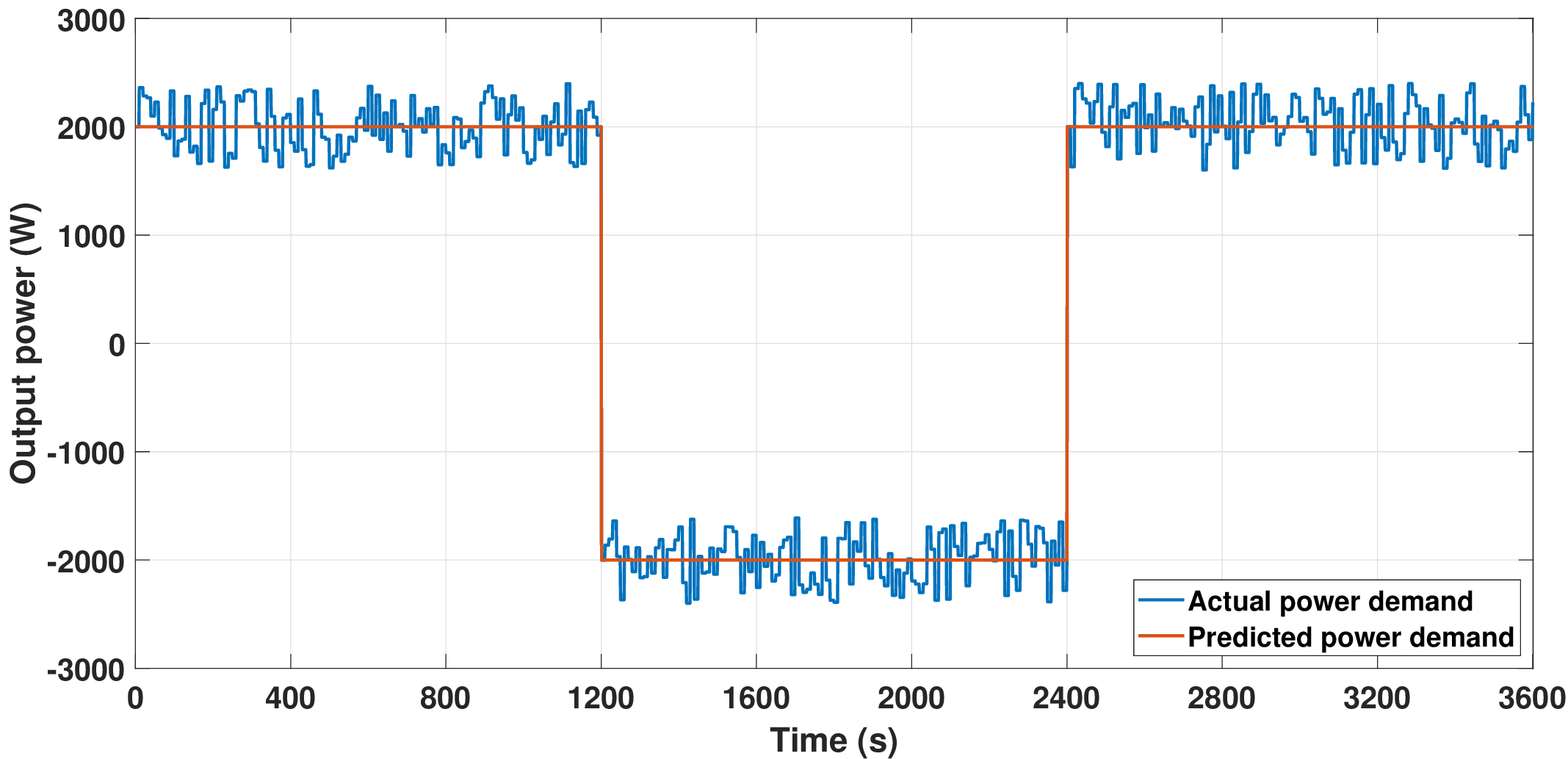}
	\caption{The output power profile.}\label{FIG_SIM_1}
\end{figure}

The ideal barrier functions in \eqref{h_Function} reflect hard constraints with zero tolerance for violations, potentially leading to less numerical stability in EnKI implementation. To mitigate the issue, we use the following as a mild compromise:
\begin{subequations}
	\begin{align}
		\psi_{g}(x) &= \begin{cases} 0 & x \leq 0, \\ \frac{1}{\alpha_g}\ln(1+\exp(\beta_g x)) & x > 0,\end{cases} \\
		\psi_{e}(x) &= \frac{1}{\alpha_e}x^{\beta_e},
	\end{align}
	\label{Barrier-2}
\end{subequations}
\hspace{-3pt}where $\alpha_{g/e}$ and $\beta_{g/e}$ are tunable parameters. Note that we apply \eqref{Barrier-2} element-wise to vectors $\bm g \left( \bm x_t, \bm \theta \right)$ and $\bm e \left( \bm x_t, \bm \theta \right)$ in \eqref{h_Function}. Additionally, $\xi_q=8$ and $\xi_\Taxi=12$ in \eqref{Functions}. The cells' initial SoC and internal resistance are heterogeneous, randomly initialized by drawing from $\mathcal{U}(0.70,0.75)$ and  $\mathcal{U}(31.3, 41.3)$ m$\Omega$, respectively. However, they are set to have the same initial temperature. The obtained simulation results are reported below.

Fig.~\ref{FIG_SIM_2} illustrates the SoC and temperature balancing performance of the proposed approach, and Fig.~\ref{FIG_SIM_3} shows the evolution of the control parameters of the parameterized control policy in \eqref{Form-1}. The two figures together highlight several distinct stages within the one-hour operation, due to the proposed OPM approach. 

\begin{itemize}
	\item {\bf 0 to about 300 seconds.} This stage prioritizes SoC balancing to overcome the initial SoC differences among the cells, as $\theta_1$, which weighs the effect of the relative SoC imbalance, is close to 1 (Fig.~\ref{FIG_SIM_3}). The cells converge in SoC (Fig.~\ref{FIG_SIM_2} (a)-(b)), while there is an increase in the temperature differences (Fig.~\ref{FIG_SIM_2} (c)-(d)). 

	\item {\bf About 300 to 1000 seconds.} This time interval mainly focuses on power loss minimization by distributing the power based on the cells' relative internal resistance while maintaining the SoC values balanced with $\theta_1, \theta_3 > \theta_2$ (Fig.~\ref{FIG_SIM_3}). The temperature differences remain bounded in this stage (Fig.~\ref{FIG_SIM_2} (c)-(d)).
	
	\item {\bf About 1000 to about 1200 seconds.} The temperature differences reach their maximum allowed threshold and it is no longer possible to simultaneously maintain SoC and temperature values balanced. The proposed approach keeps the temperature values balanced with $\theta_2 \gg \theta_1, \theta_3$ (Fig.~\ref{FIG_SIM_3}) while allowing SoC values to be slightly unbalanced (Fig.~\ref{FIG_SIM_2} (a)-(b)).
	
	\item {\bf About 1200 to about 1700 seconds.} The BESS switches from discharging to charging and the emphasis of this stage is on power loss minimization, with $\theta_3 \gg \theta_1,\theta_2$ (Fig.~\ref{FIG_SIM_3}), following the achievement of the SoC and temperature balancing in the preceding stage. This emphasis also leads to an increase in the SoC and temperature differences during this stage (Fig.~\ref{FIG_SIM_2}), but the differences remain bounded. 
	
	\item {\bf About 1700 to 2400 seconds.} The BESS has reached a balance among all three objectives at the end of the previous stage, and thus puts an almost equal emphasis in this stage (Fig.~\ref{FIG_SIM_3}). 
	
	\item {\bf 2400 to about 2900 seconds.} The BESS enters the discharging mode again. This shift leads the BESS to enforce power loss minimization and without violating the SoC and temperature balancing constraints (Fig.~\ref{FIG_SIM_2}). As such, $\theta_3 \gg \theta_1,\theta_2$ (Fig.~\ref{FIG_SIM_3}). 

	\item {\bf About 2900 to about 3600 seconds.} During this stage, the cells have been well balanced in SoC and temperature (Fig.~\ref{FIG_SIM_2}) and the BESS balances the three objectives to maintain this condition.
\end{itemize}

\begin{figure*}[t]
	    \centering
    \subfloat[\centering ]{{\includegraphics[trim={0cm 0cm 0cm 0cm},clip,width=8.5cm]{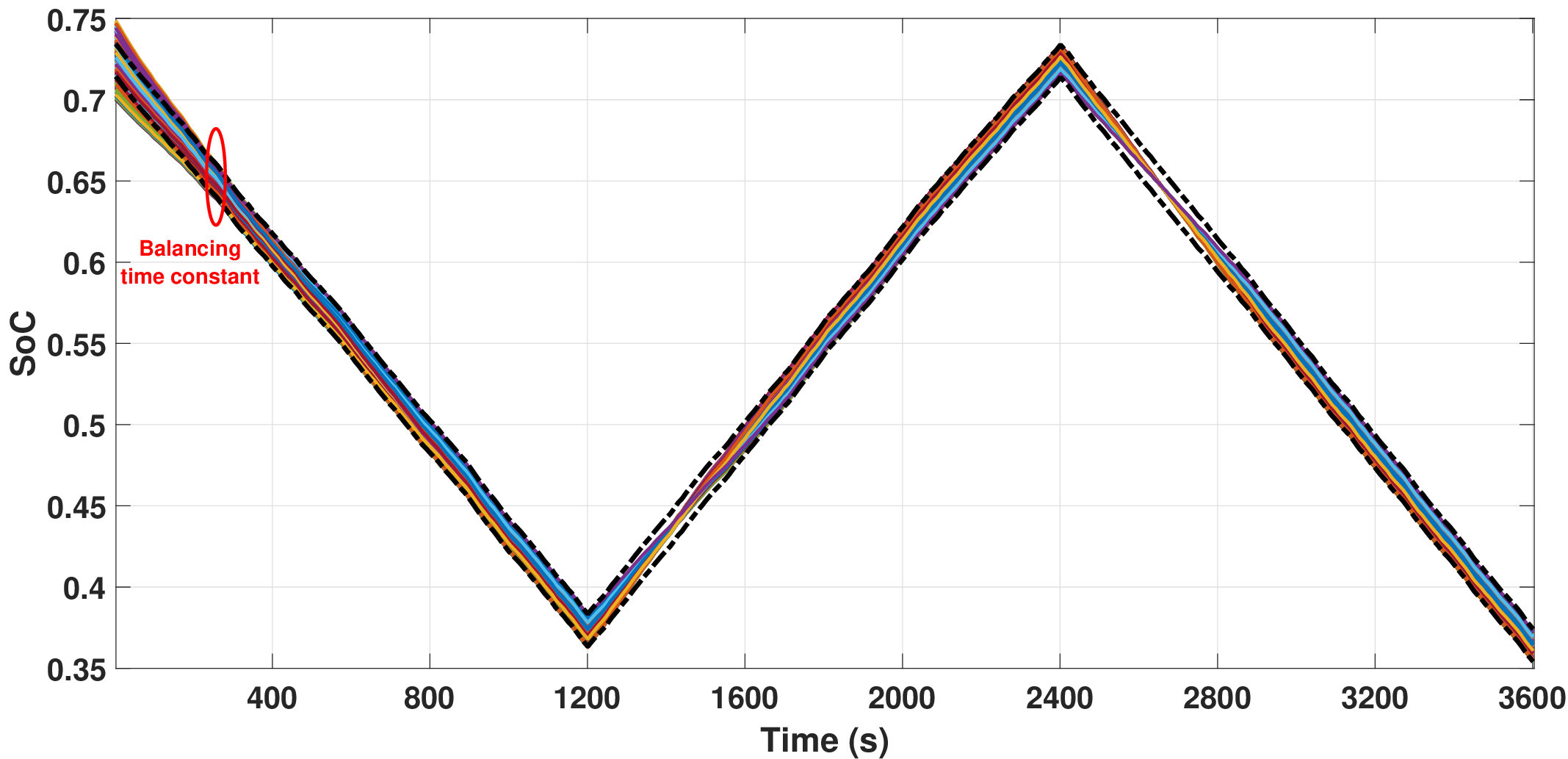} }}
	\,
    \subfloat[\centering ]{{\includegraphics[trim={0cm 0cm 0cm 0cm},clip,width=8.5cm]{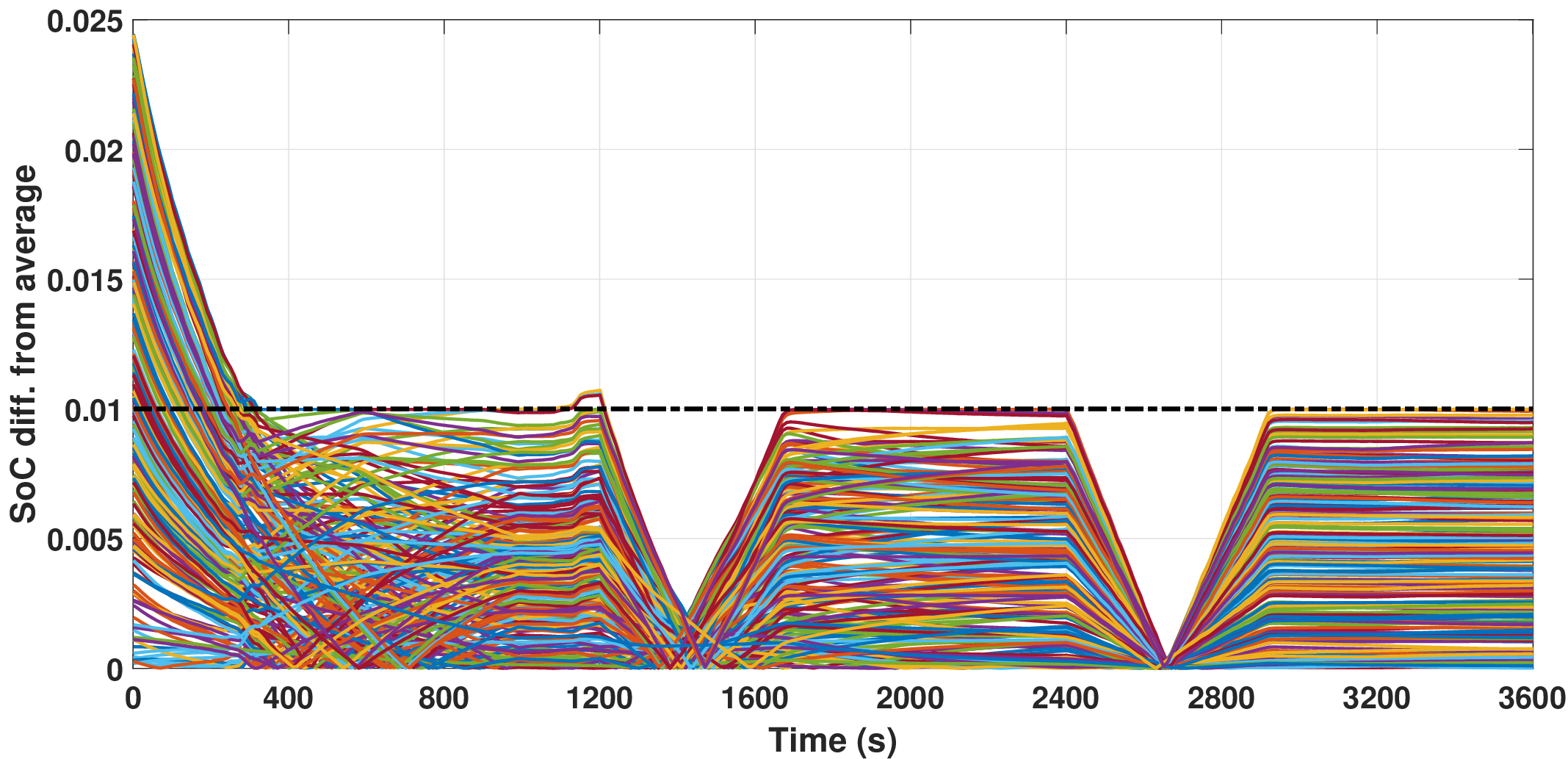} }}
    \,
    \subfloat[\centering ]{{\includegraphics[trim={0cm 0cm 0cm 0cm},clip,width=8.5cm]{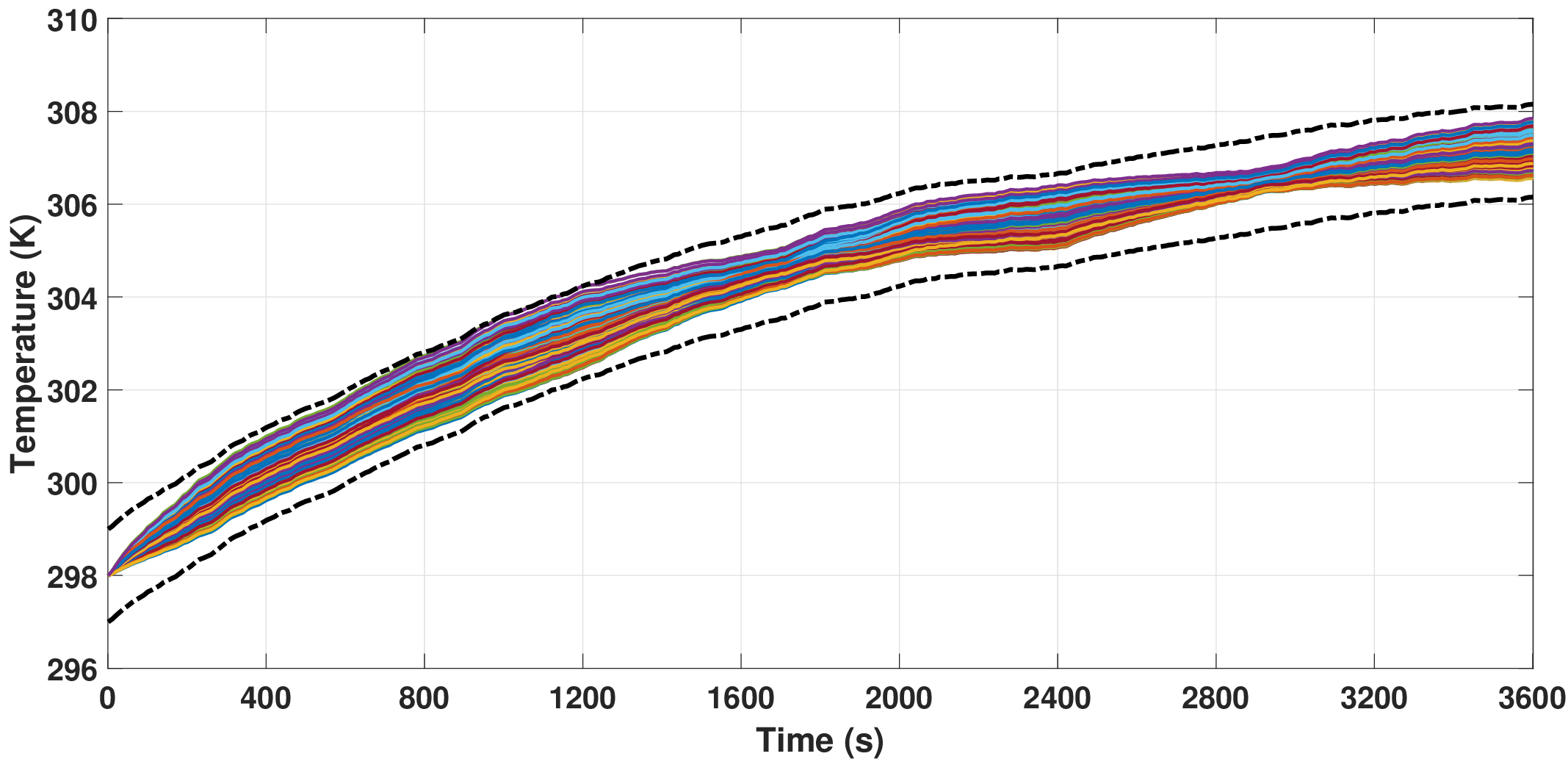} }}
    \,
    \subfloat[\centering ]{{\includegraphics[trim={0cm 0cm 0cm 0cm},clip,width=8.5cm]{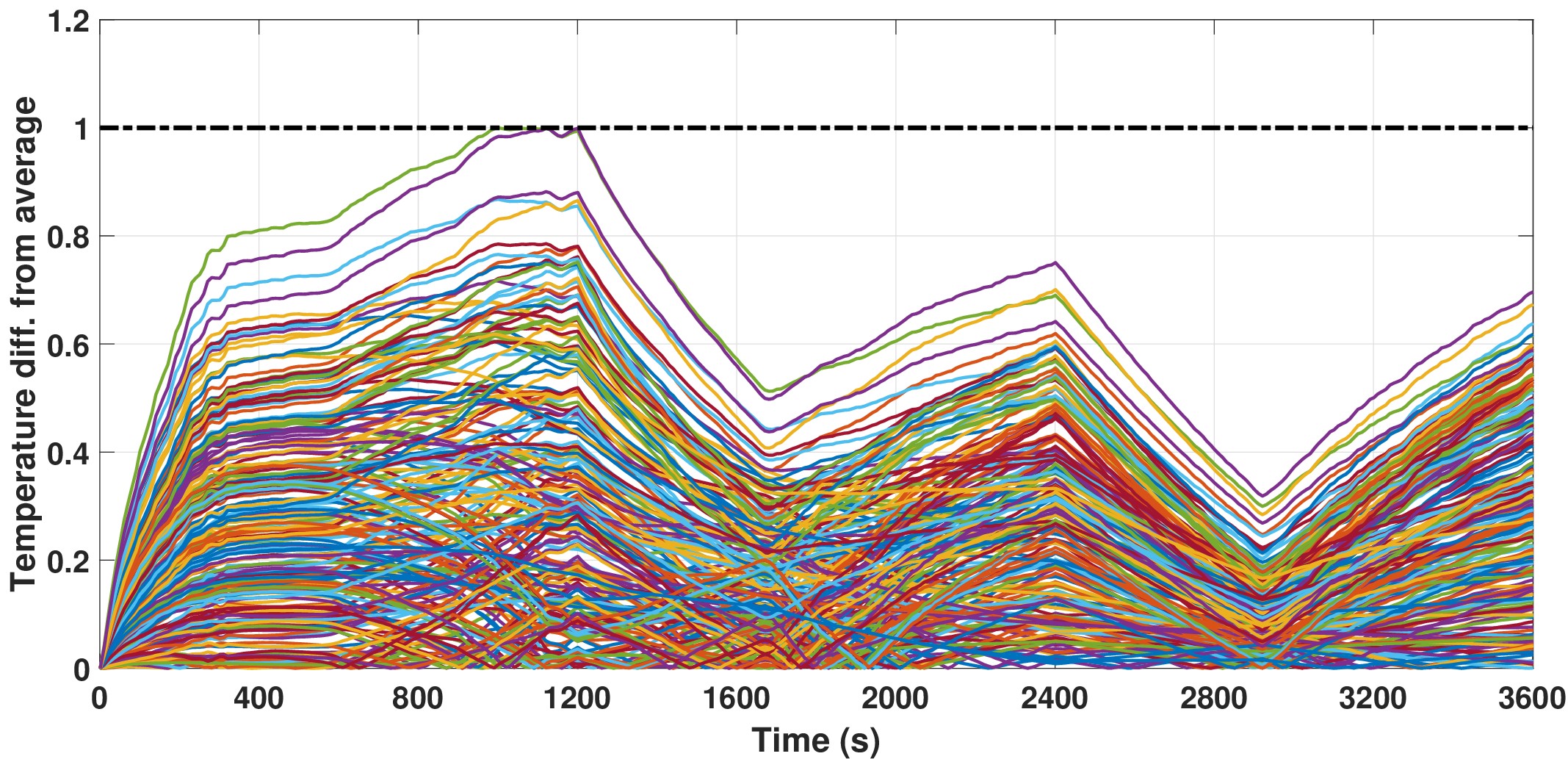} }}
    \caption{Simulation results of the SoC and temperature balancing. (a) The SoC of the cells. (b) The SoC difference of the cells from the average. (c) The temperature of the cells. (d) The temperature difference of the cells from the average.}
    \label{FIG_SIM_2}
\end{figure*}

\begin{figure}[!t]\centering
	\includegraphics[trim={0cm 0cm 0cm 0cm},clip,width=\linewidth]{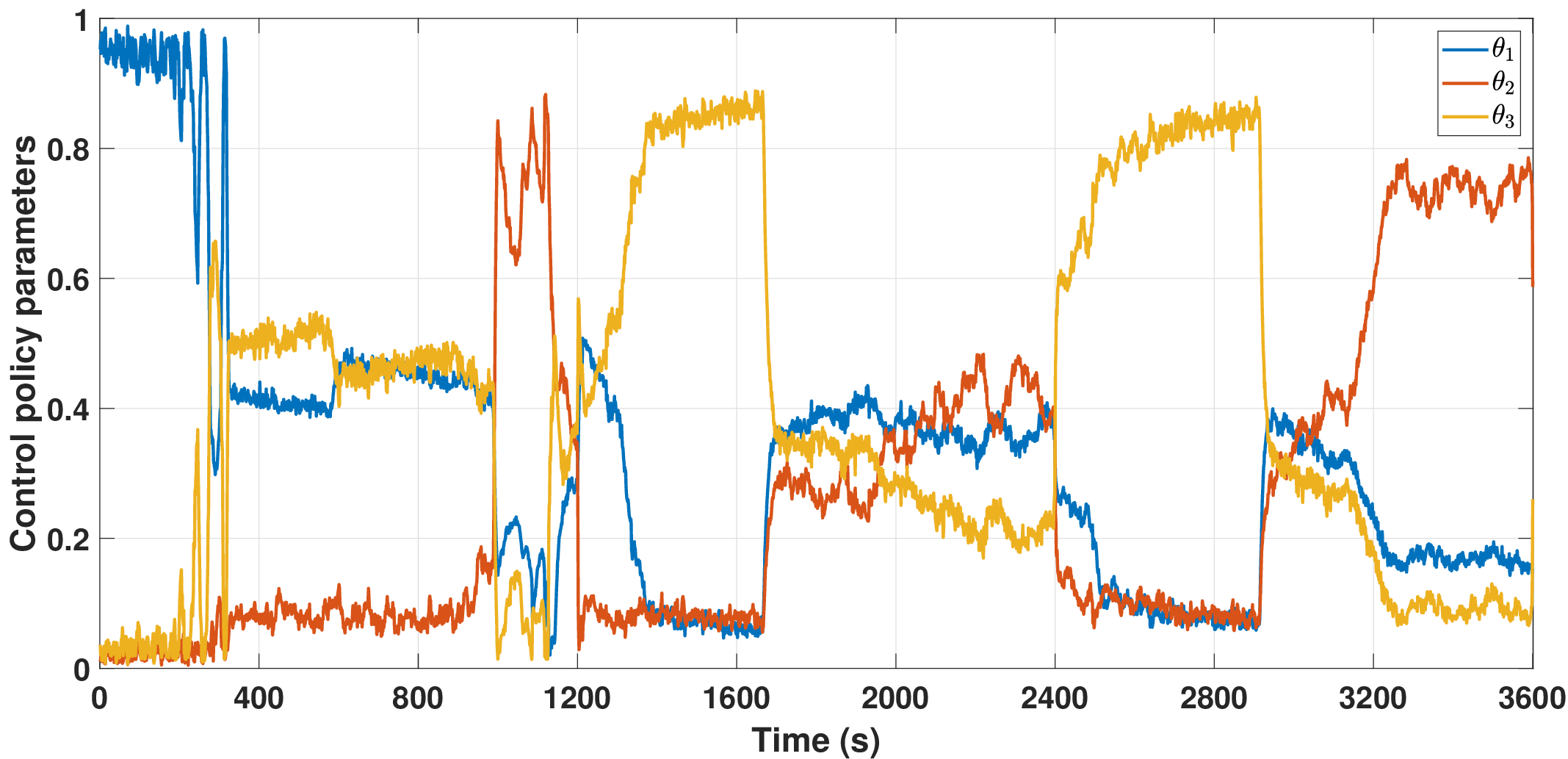}
	\caption{Control policy parameters.}\label{FIG_SIM_3}
\end{figure}

Fig.~\ref{FIG_SIM_4} illustrates the PSR profiles. As is shown, the cells' PSR are adjusted by the proposed OPM approach based on their operating conditions. In general, the cells have largely different PSR when the BESS must focus more on SoC balancing or power loss minimization. But the difference becomes much smaller when the BESS attains a balance among all three objectives. For instance, the PSR converges toward about $1/n$ within the stage of about 2900 to 3600 seconds.

\begin{figure}[!t]\centering
	\includegraphics[trim={0cm 0cm 0cm 0cm},clip,width=\linewidth]{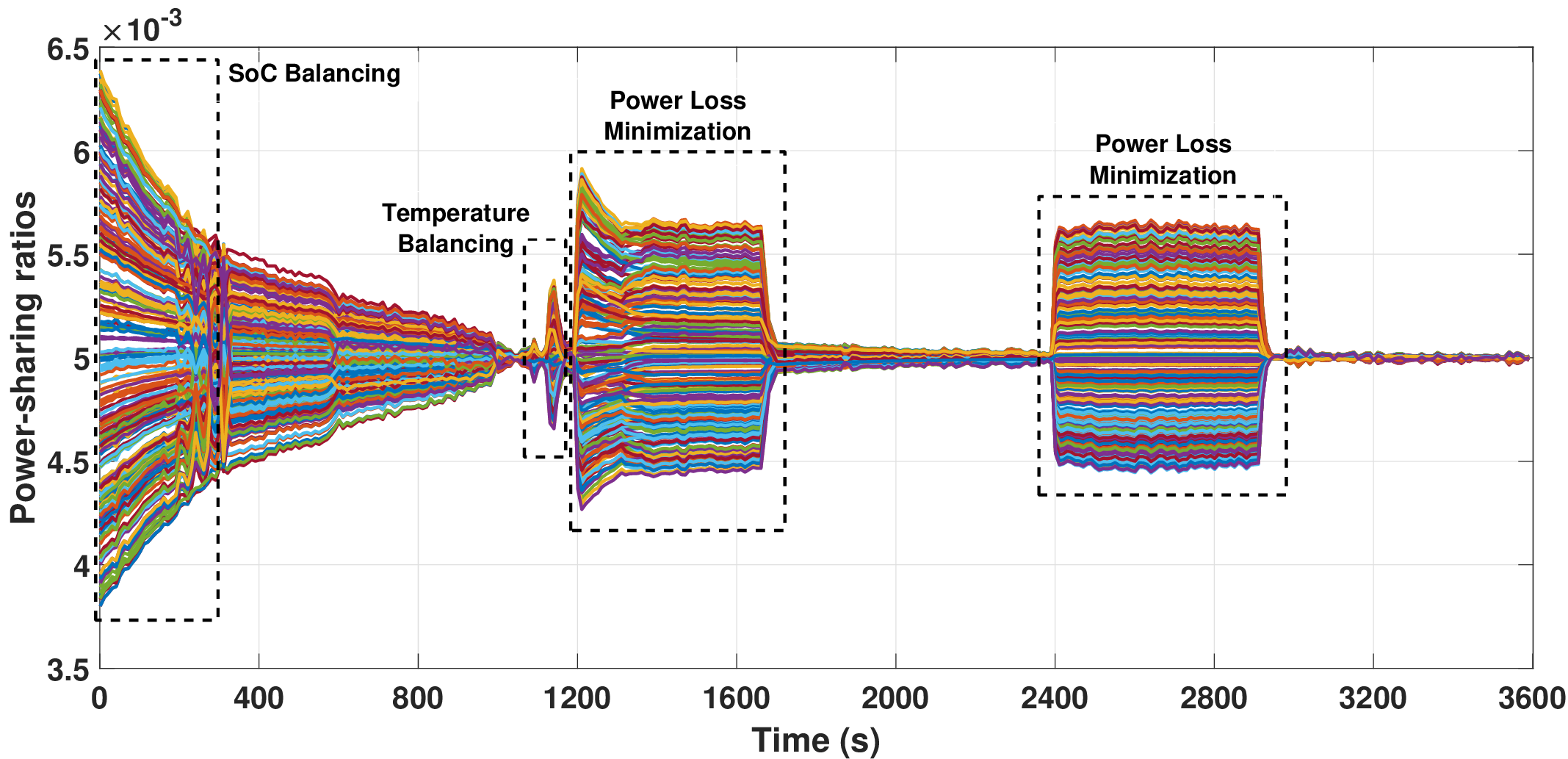}
	\caption{The cells' power sharing ratios.}\label{FIG_SIM_4}
\end{figure}

%After computing the utilization factors, we proceed to update them in the low-level controller, as depicted in Fig.\ref{Low-Level}. The low-level controller then generates an error signal by comparing the BESS output voltage with its reference value. The PI controller, with coefficients set to $K_P=2$ and $K_I=5$ minimizes this error signal by producing a residual power signal, denoted as $\tilde{P}_{\textrm{out}}$, as illustrated in Fig.\ref{FIG_SIM_5}. This residual power signal captures the uncertainty in our output power predictions. Subsequently, the low-level controller combines this residual power with the predicted output power to compute the reference power for each cell based on their respective utilization factors. By doing so, the imperfect output power predictions are circumvented while maintaining optimality through the utilization factors. 
%
%\begin{figure}[!t]\centering
%	\includegraphics[trim={2.4cm 0 2.75cm 1cm},clip,width=\linewidth]{Figures/ResidualPower.eps}
%	\caption{The residual power generated by the PI controller $\tilde{P}_{\textrm{out}}$.}\label{FIG_SIM_5}
%\end{figure}

Further, we conduct various simulations to compare the proposed OPM approach with cell-level power optimization methods \cite{2019-TVT-CR,2023-TTE-FA} in terms of computational efficiency and scalability. The simulations involve different cell number, horizon lengths, and particle numbers, while the method \cite{2023-TTE-FA} is selected as the benchmark. We run the simulations on a workstation with a 3.5 GHz Intel Core i9-10920X CPU and 128 GB of RAM and using Matlab. The results are presented in Table \ref{Table: Computation-Comp} and visualized in Fig.~\ref{Computation-Comp}, revealing several key findings. First, our approach reduces the computation time by considerable margins exceeding 90\% in all scenarios. This reduction remains substantial even as the horizon length or the number of particles increases. Second, our proposed approach shows strong computational scalability, making it possible to manage a large number of cells at low computational costs. By contrast, the cell-level optimization method scales poorly with an almost exponential increase in computation as the number of cells grows.

\begin{table}[t!]\centering
\caption{Numerical comparison of the proposed approach and cell-level optimization}
\resizebox{0.48\textwidth}{!}{
 \begin{tabular}{>{\color{black}}c | >{\color{black}}c | >{\color{black}}l  >{\color{black}}c  >{\color{black}}c}
\toprule
\makecell[c]{Cell Number ($n$)} & \makecell[c]{Horizon \\ Length} & \makecell[c]{Method} & \makecell[c]{ Average \\ Computation \\ Time (s)} & \makecell[c]{Relative \\ Computation \\ Time Reduction\\ (\%)}  \\
\midrule
\multirow{10}{0.5cm}{\makecell[c]{50}} & \multirow{3}{0.5cm}{\makecell[c]{5}} & {\bf Cell-level optimization} & 12.56 & --- \\
& & {\bf Proposed with 50 particles} & 0.17 & 98.64  \\
& & {\bf Proposed with 100 particles} & 0.26 & 97.92  \\
\cmidrule(l){2-5}
& \multirow{3}{0.5cm}{\makecell[c]{10}} & {\bf Cell-level optimization} & 29.02 & ---  \\
& & {\bf Proposed with 50 particles} & 0.61 & 97.89  \\
& & {\bf Proposed with 100 particles} & 0.79 & 97.27  \\
\cmidrule(l){2-5}
& \multirow{3}{0.5cm}{\makecell[c]{15}} & {\bf Cell-level optimization} & 40.76 & ---  \\
& & {\bf Proposed with 50 particles} & 1.36 & 96.66  \\
& & {\bf Proposed with 100 particles} & 2.15 & 94.72  \\
\cmidrule(l){1-5}

\multirow{10}{0.5cm}{\makecell[c]{100}} & \multirow{3}{0.5cm}{\makecell[c]{5}} & {\bf Cell-level optimization} & 28.56 & --- \\
& & {\bf Proposed with 50 particles} & 0.43 & 98.49  \\
& & {\bf Proposed with 100 particles} & 0.50 & 98.24 \\
\cmidrule(l){2-5}
& \multirow{3}{0.5cm}{\makecell[c]{10}} & {\bf Cell-level optimization} & 72.77 & ---  \\
& & {\bf Proposed with 50 particles} & 1.49 & 97.95  \\
& & {\bf Proposed with 100 particles} & 1.73 & 97.62  \\
\cmidrule(l){2-5}
& \multirow{3}{0.5cm}{\makecell[c]{15}} & {\bf Cell-level optimization} & 139.93 & ---  \\
& & {\bf Proposed with 50 particles} & 3.05 & 97.82  \\
& & {\bf Proposed with 100 particles} & 3.63 & 94.40 \\
\cmidrule(l){1-5}

\multirow{10}{0.5cm}{\makecell[c]{200}} & \multirow{3}{0.5cm}{\makecell[c]{5}} & {\bf Cell-level optimization} & 101.48 & --- \\
& & {\bf Proposed with 50 particles} & 1.01 & 99.00  \\
& & {\bf Proposed with 100 particles} & 1.33 & 98.68  \\
\cmidrule(l){2-5}
& \multirow{3}{0.5cm}{\makecell[c]{10}} & {\bf Cell-level optimization} & 274.97 & ---  \\
& & {\bf Proposed with 50 particles} & 4.62 & 98.31  \\
& & {\bf Proposed with 100 particles} & 5.90 & 97.85 \\
\cmidrule(l){2-5}
& \multirow{3}{0.5cm}{\makecell[c]{15}} & {\bf Cell-level optimization} & 423.76 & ---  \\
& & {\bf Proposed with 50 particles} & 10.56 & 97.50  \\
& & {\bf Proposed with 100 particles} & 11.88 & 97.19  \\
\cmidrule(l){1-5}
\end{tabular}
 }
\label{Table: Computation-Comp}
\end{table}

\begin{figure}[!t]\centering
	\includegraphics[trim={0cm 0cm 1.5cm 0cm},clip,width=8cm]{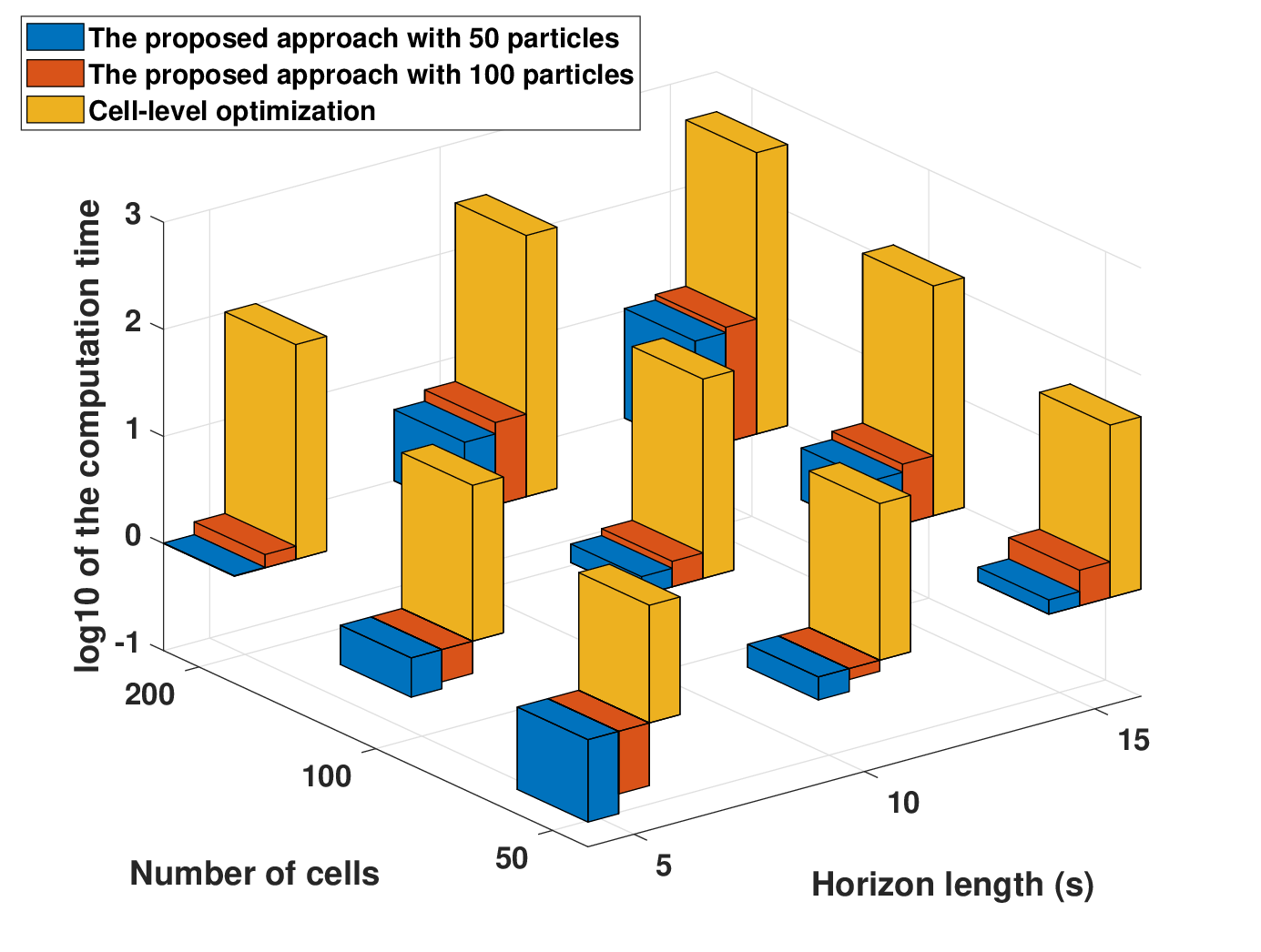}
	\caption{Comparative analysis of computation times: the proposed approach with 50 and 100 particles versus the cell-level optimization.}\label{Computation-Comp}
\end{figure}
\section{Experimental Results}
In this section, we validate the practical performance of the proposed optimal power management approach using a lab-scale prototype. Fig.~\ref{FIG_EXP_1} illustrates the experimental setup based on a 20-cell 4s5p battery pack. Table \ref{TABLE_EXP} provides the specifications for the key components of the battery pack. We employ K-type thermocouples to measure the cells' surface temperatures. Temperature and voltage measurements are collected using a National Instruments PCIe-6321 DAQ board with LabView. We use Matlab to solve the estimation problem, while the local-level controller is implemented in Matlab/Simulink. We interface the local-level controllers with the DC/DC converters through the DSP TMS320F28335. The battery pack is discharged with a nominal output power of 90 W, and the terminal voltage is set at 30 V. The experiment lasts for 30 minutes with the proposed OPM approach running every 30 seconds.

\begin{figure}[!t]\centering
	\includegraphics[trim={0cm 0cm 0cm 0cm},clip,width=\linewidth]{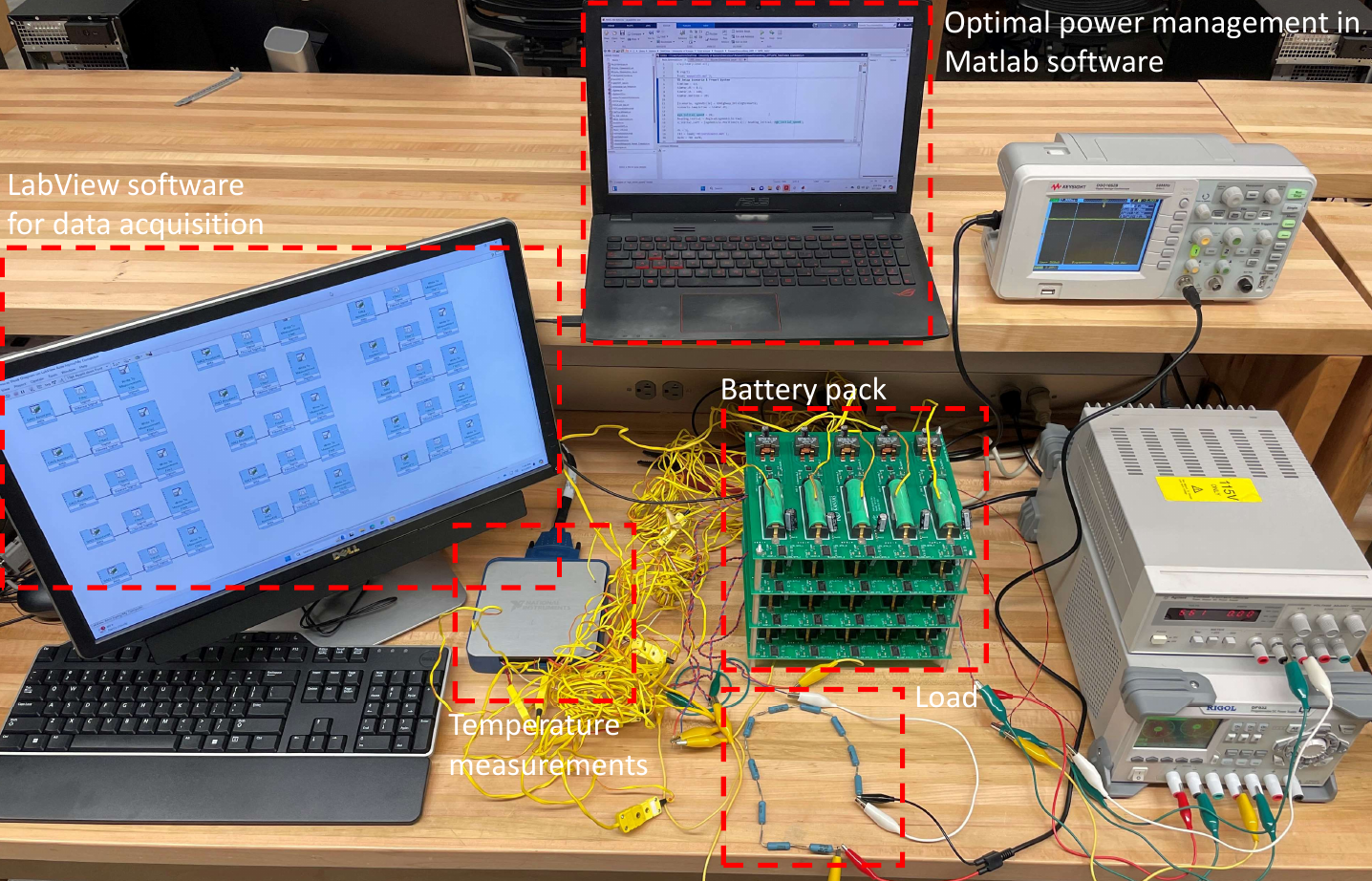}
	\caption{The experimental setup.}\label{FIG_EXP_1}
\end{figure}

%\begin{figure}[!t]
%    \centering
%    \subfloat[\centering ]{{\includegraphics[trim={0cm 1.5cm 0cm 0cm},clip,width=\linewidth]{Figures/EXP-Setup.pdf} }}
%    \;
%   \vspace{0.5em} \subfloat[\centering ]{{\includegraphics[trim={0cm 0cm 0cm 0cm},clip,width=5cm]{Figures/Circuit.pdf} }}
%    \;
%    \caption{The lab-scale prototype. (a) The experimental setup. (b) The circuit.}
%    \label{FIG_EXP_1}
%\end{figure}

The experiment starts with cells' SoC values ranging from 70\% to 80\%, while maintaining uniform temperatures at 22.3\degree C. We then increase the output power to 120 W twice to evaluate the performance of the proposed approach under output power uncertainty. The experimental results are presented in Figs.~\ref{FIG_EXP_2}-\ref{FIG_EXP_6}.

\begin{table}[!b]
	\renewcommand{\arraystretch}{1.3}
	\caption{List of Key Components in the Experimental Setup}
	\centering
	\label{TABLE_EXP}
	\centering
	{
		\begin{tabular}{l l}
			\hline\hline \\ [-3mm]
			\multicolumn{1}{c}{Device} & \multicolumn{1}{c}{Model (Value)}\\[1.6ex] \hline
			MOSFET & CSD86356Q5D \\
			Gate driver & TPS28225 \\
			Inductor & SER2915H-333KL (33 $\mu$H) \\
			Capacitor & (10 $\mu$F) \\
			Local controller \qquad \qquad \qquad & STM8S003F3P6 \\
			Main controller & TMS320f28335 \\
			Battery cell & Samsung INR18650-25R \\
			\hline\hline
		\end{tabular}
	}
\end{table}

\begin{figure*}[t]
	    \centering
    \subfloat[\centering ]{{\includegraphics[trim={2cm 0 2.75cm 1cm},clip,width=8.8cm]{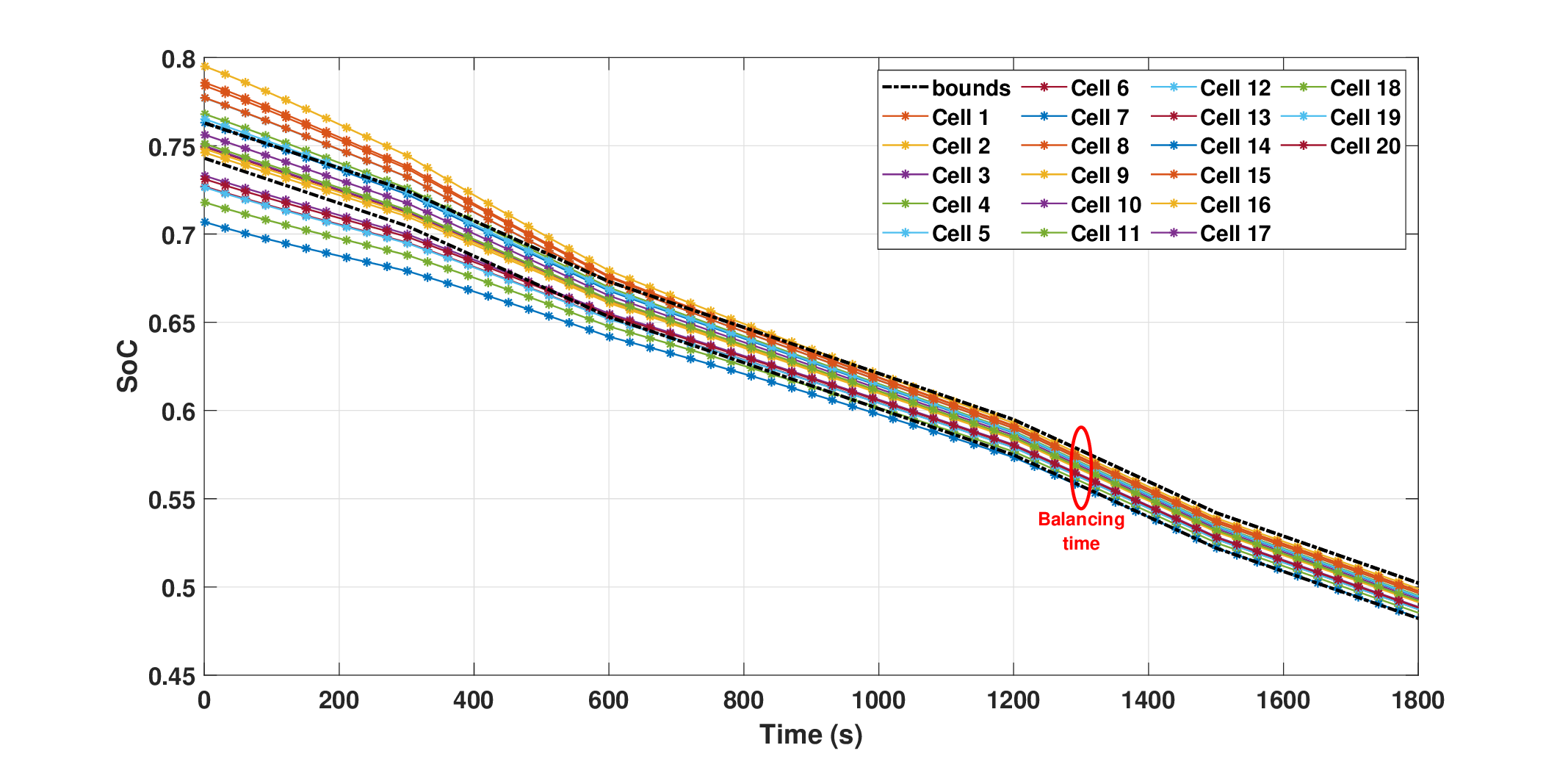} }}
    \,
    \subfloat[\centering ]{{\includegraphics[trim={2cm 0 2.75cm 1cm},clip,width=8.8cm]{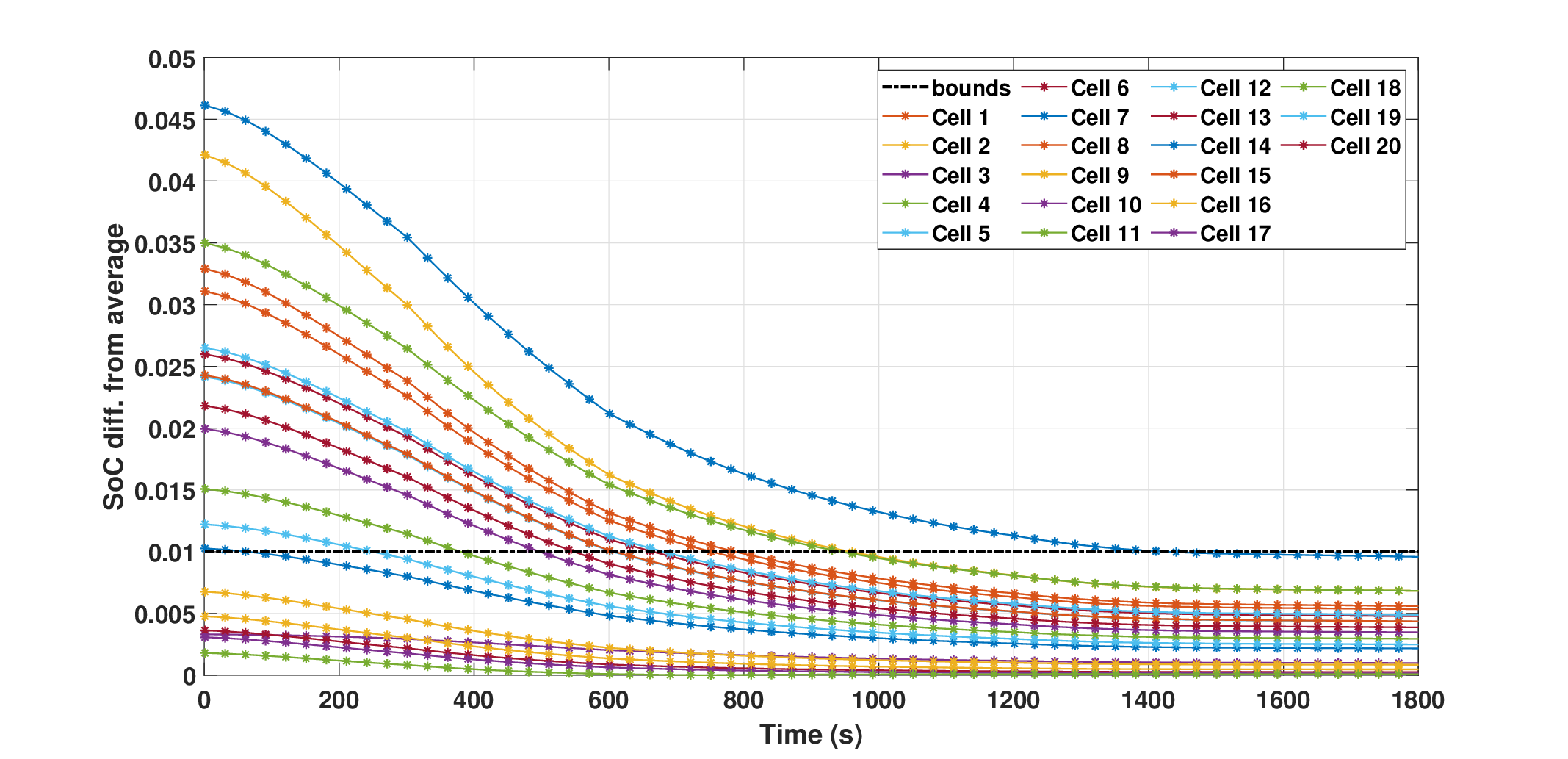} }}
    \,
    \subfloat[\centering ]{{\includegraphics[trim={2cm 0 2.75cm 1cm},clip,width=8.8cm]{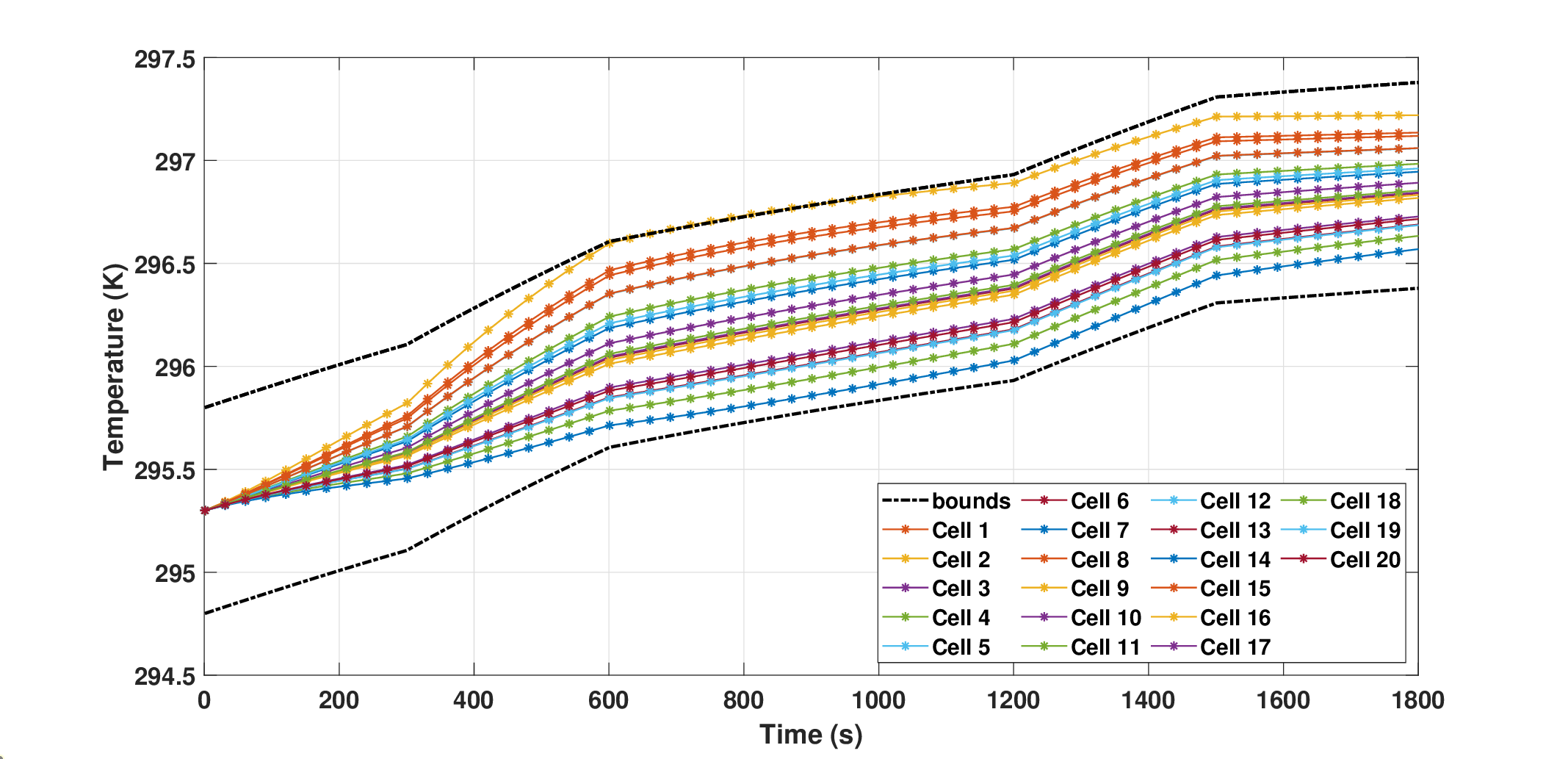} }}
    \,
    \subfloat[\centering ]{{\includegraphics[trim={2cm 0 2.75cm 1cm},clip,width=8.8cm]{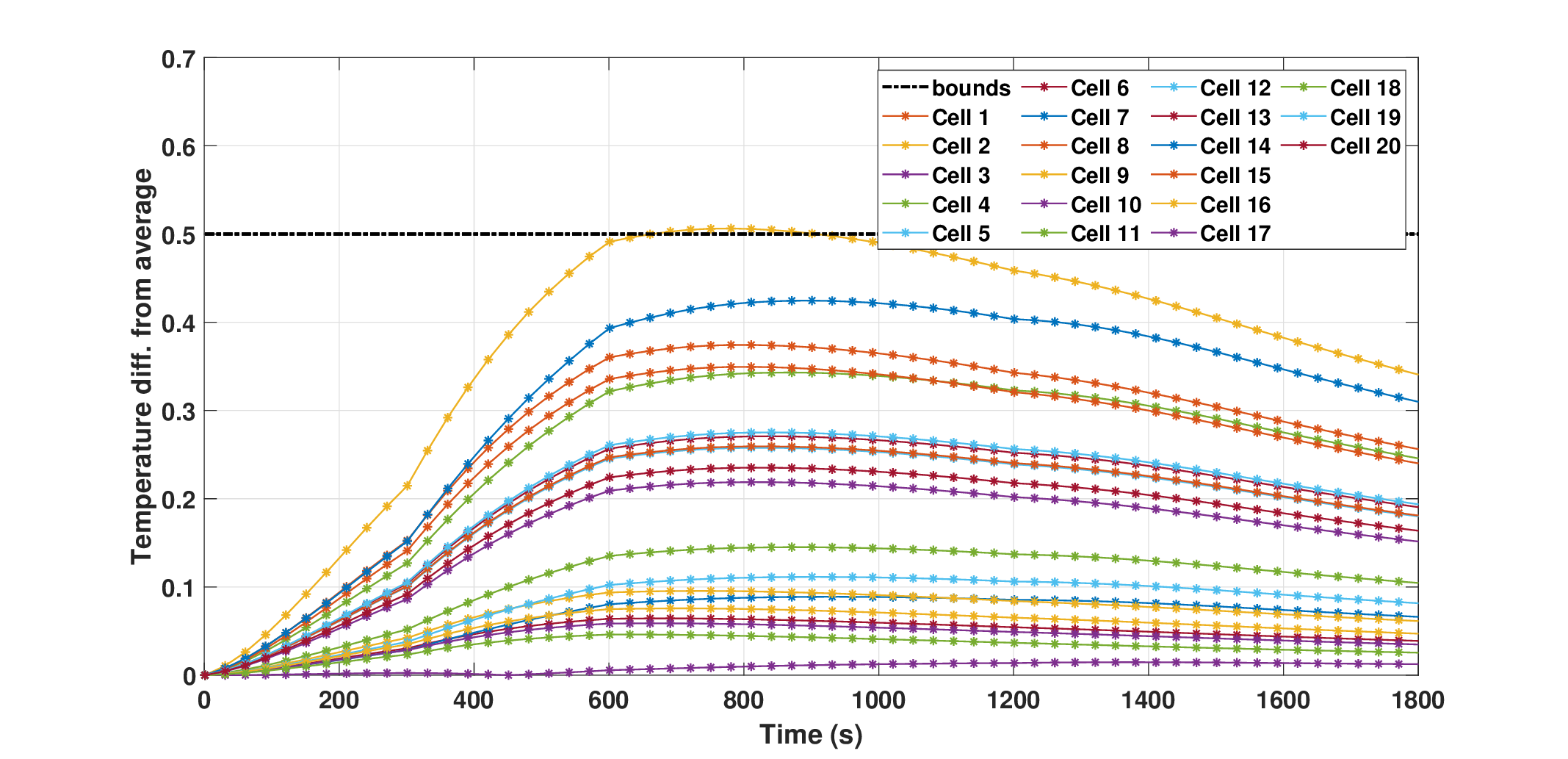} }}
    \caption{Experimental results of the SoC and temperature balancing. (a) The SoC of the cells. (b) The SoC difference of the cells from the average. (c) The temperature of the cells. (d) The temperature difference of the cells from the average.}
    \label{FIG_EXP_2}
\end{figure*}

Figs.~\ref{FIG_EXP_2} (a)-(b) show the cells' SoC values and their deviations from the average value, respectively. Initially, the cells' SoC values do not fall within the desired bounds. However, after approximately 1300 seconds, the proposed optimal power management effectively brings the cells' SoC values within the tolerable range. To delve deeper into the results, we present the control parameters of the parameterized control policy in Fig.~\ref{FIG_EXP_3}. The proposed approach initially increases $\theta_1$ to prioritize SoC balancing, leading to the growing difference in the cells' PSR values as shown in Fig.~\ref{FIG_EXP_4}. After about 1300 seconds, $\theta_1$ decreases but $\theta_2$ and $\theta_3$ increase, implying that the BESS focuses more on power loss minimization and temperature balancing (Fig.~\ref{FIG_EXP_3}). Figs.~\ref{FIG_EXP_2} (c)-(d) depict the cells' temperature values and their deviations from the average temperature. Although the cells' temperature values start to diverge in the initial stage as the proposed approach unevenly distributes the output power demand, they consistently remain within the desired balancing range throughout the experiment.

Examining the PSR in Fig.~\ref{FIG_EXP_4}, we observe an initially uneven power distribution among the cells due to the initial SoC imbalance. However, as the experiment progresses and the cells become balanced, the PSR converges to $1/n$ with $n = 20$, indicating a uniform power distribution among the cells.

\begin{figure}[!t]\centering
	\includegraphics[trim={2.4cm 0 2.75cm 1cm},clip,width=\linewidth]{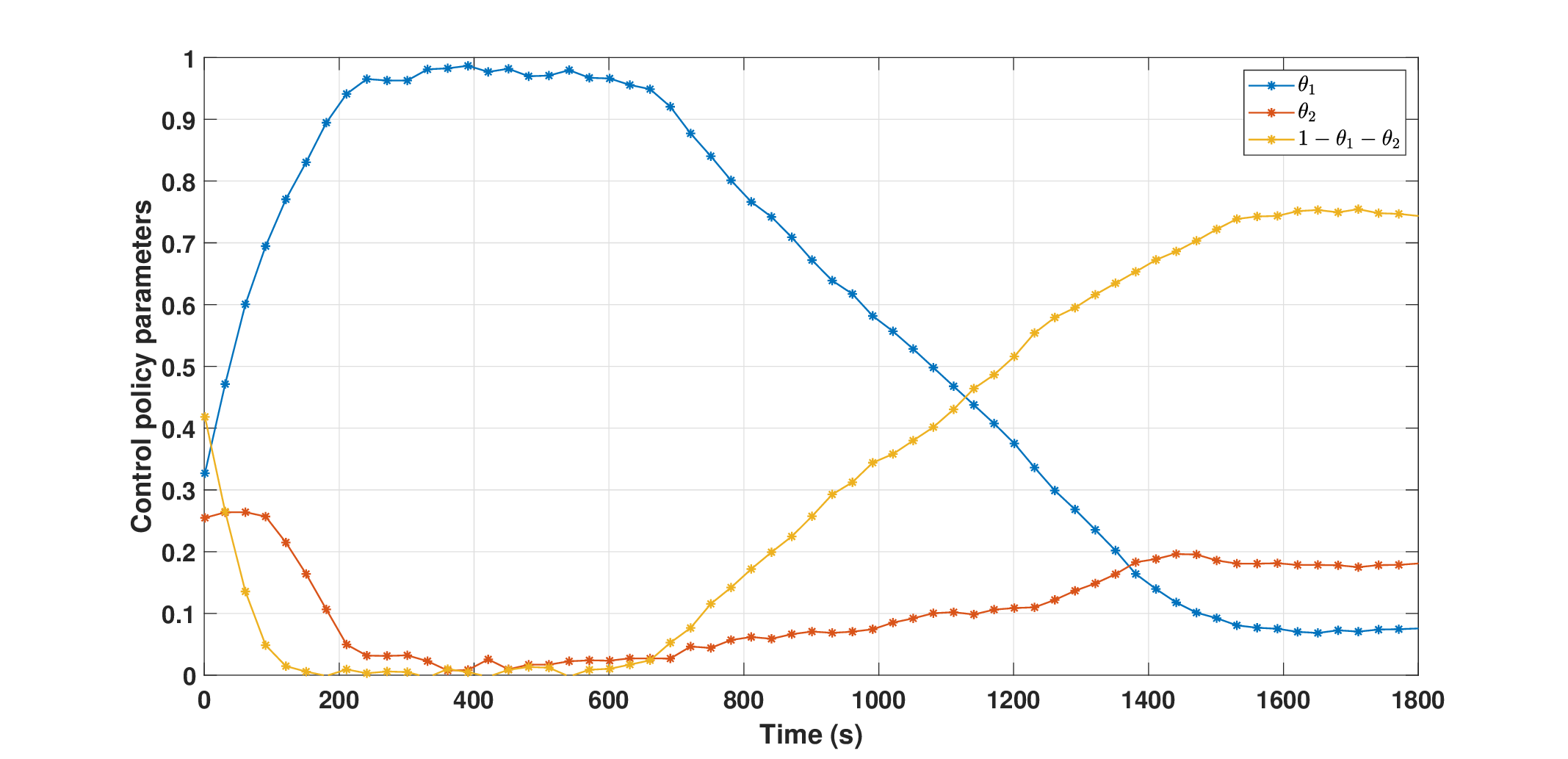}
	\caption{Control input parameters.}\label{FIG_EXP_3}
\end{figure}

\begin{figure}[!t]\centering
	\includegraphics[trim={2.4cm 0 2.75cm 1cm},clip,width=\linewidth]{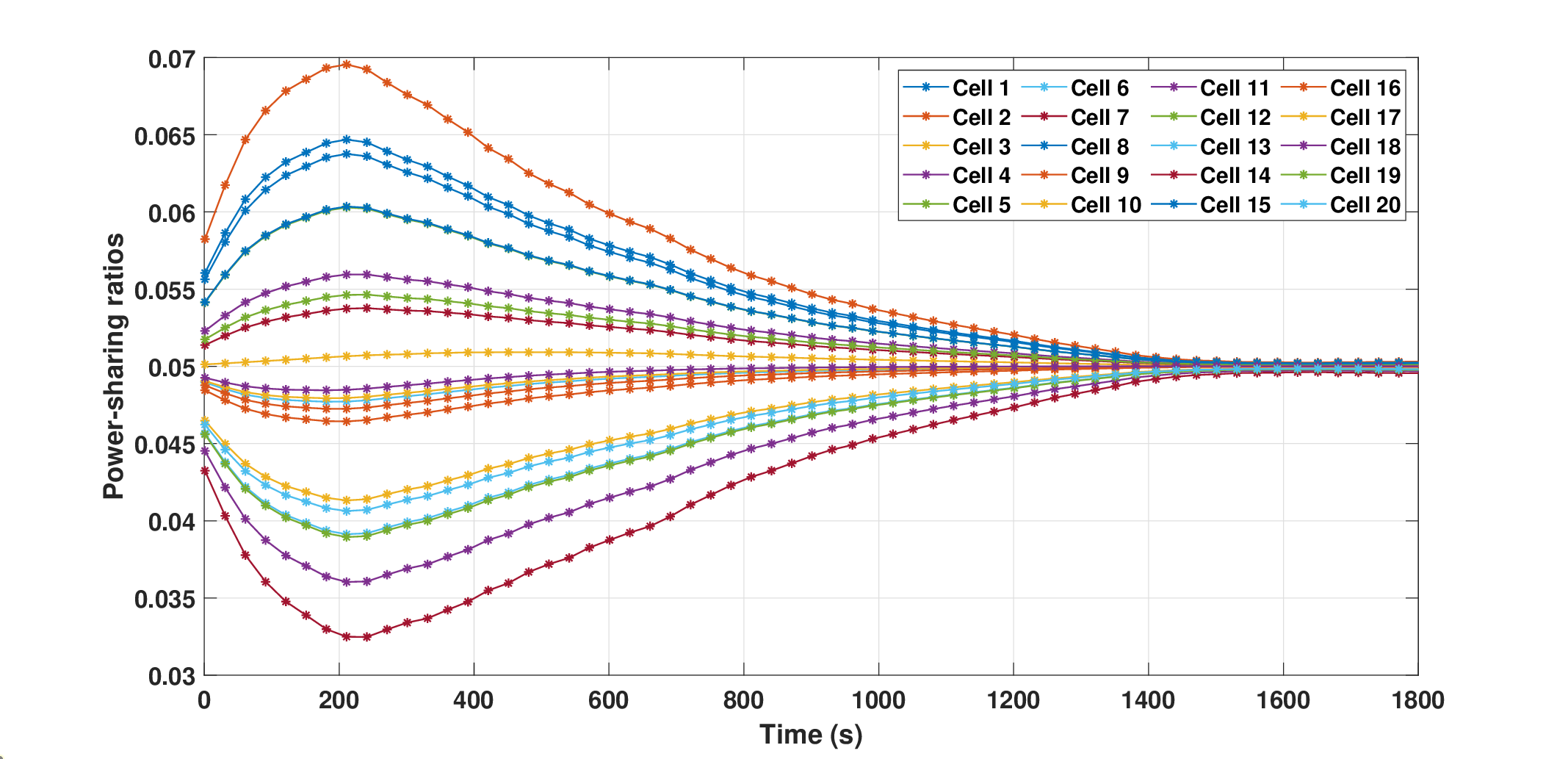}
	\caption{The cells' PSR.}\label{FIG_EXP_4}
\end{figure}

To evaluate the proposed approach's performance under output power uncertainties, we increase the actual output power from 90 W to 120 W at two time instants: the 300-th and 1200-th second. Fig.~\ref{FIG_EXP_5} shows the reference power for each cell generated by the local controller based on the PSR values. While the control policy and PSR computation are still based on the output power of 90 W, the reference power for the cells increases to adequately meet the load demand. It is important to notice that this increase varies among the cells due to the incorporation of optimal PSR. A similar trend occurs when the output power increases when the actual power changes again at the 1200-th second.

The BESS terminal voltage is also measured by a 10 kHz analog-digital converter in Fig.~\ref{FIG_EXP_6} at the 300-th second when the output power rises from 90 W to 120 W. The terminal voltage experiences a dip, but soon returns to 30 V due to the compensation by the local controller. While regulating the terminal voltage, the local controller also generates the residual power to make up for the 30 W difference between the predicted and actual power demands.

\begin{figure}[!t]\centering
	\includegraphics[trim={2.4cm 0 2.75cm 1cm},clip,width=\linewidth]{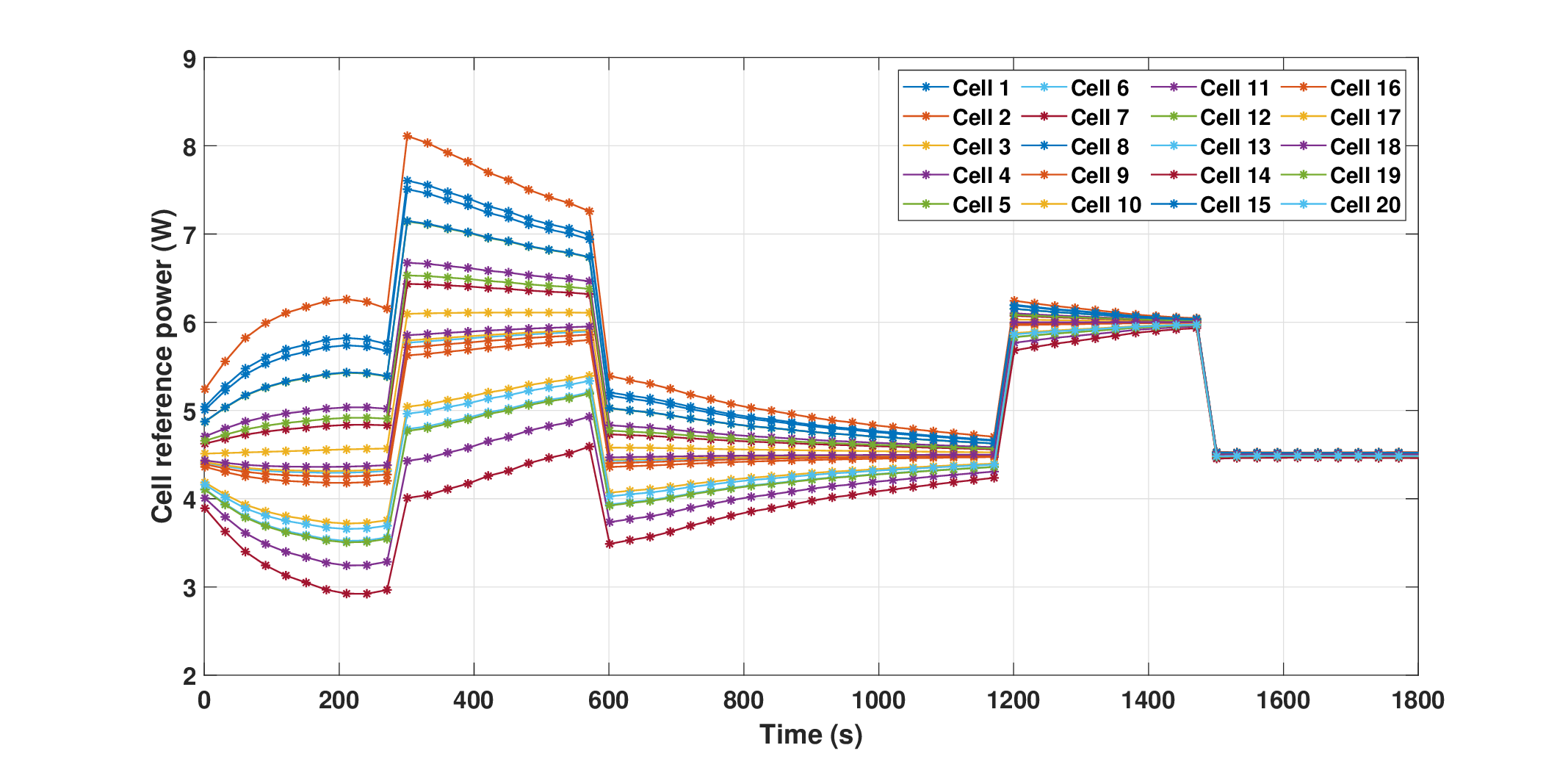}
	\caption{The cells' reference power.}\label{FIG_EXP_5}
\end{figure}

\begin{figure}[!t]\centering
	\includegraphics[trim={2.4cm 0 2.75cm 1cm},clip,width=\linewidth]{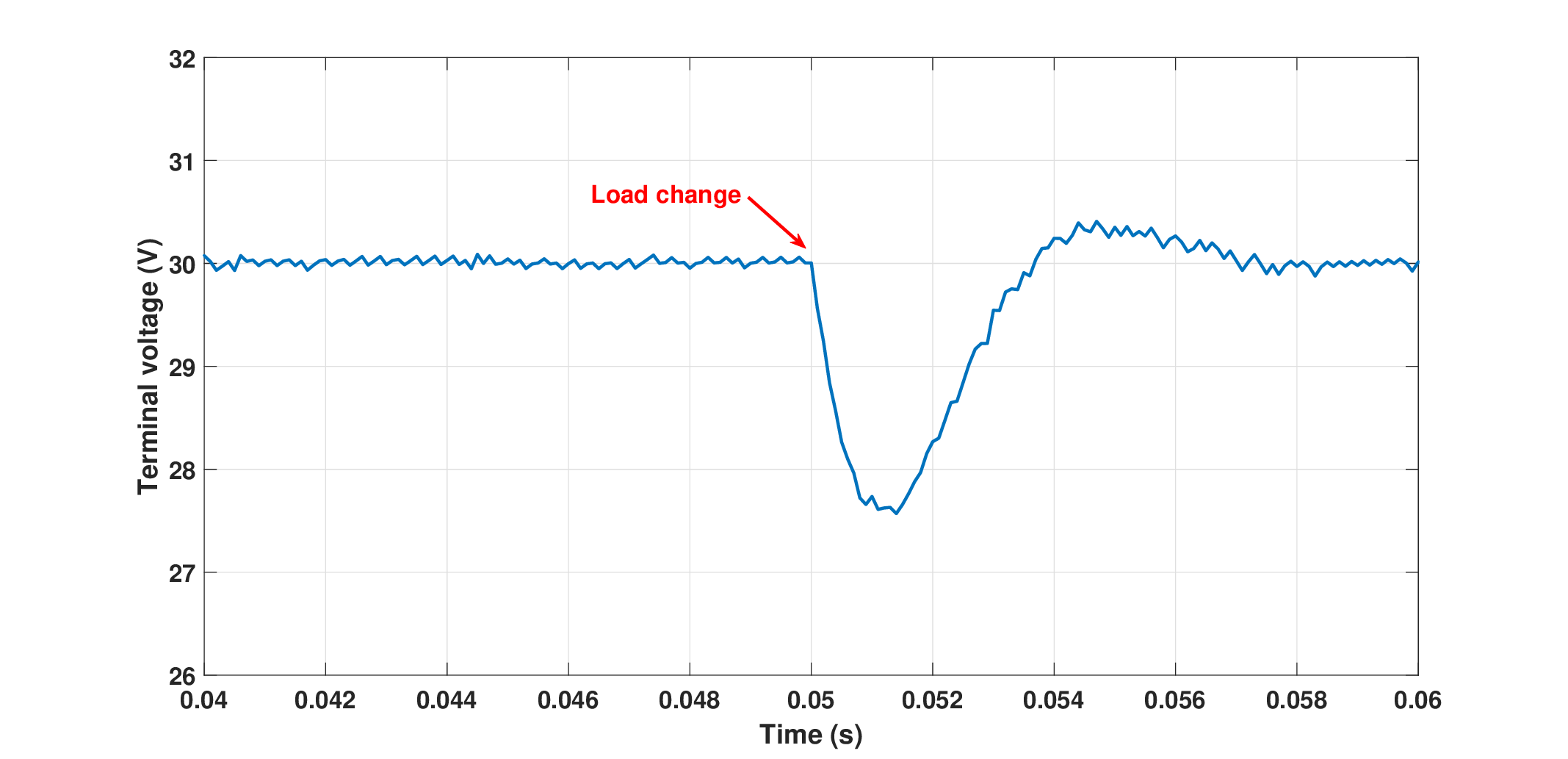}
	\caption{The BESS terminal voltage when the load increases abruptly.}\label{FIG_EXP_6}
\end{figure}
\section{Conclusion}
OPM is the cornerstone for the efficient and safe operation of large-scale BESS emerging in various sectors. This problem, however, is computationally challenging, as it requires to handle high-dimensional nonlinear nonconvex optimization. Further, OPM can experience performance loss in practical implementation because it relies on the prediction of output power demands, which are inherently uncertain. To address these challenges, we develop a new framework characterized by two key elements. First, we introduce the concept of PSR to represent each cell's contribution to the BESS operation. This ratio-based power allocation provides a flexible and practical way to meeting uncertain power demands. Second, we set up a PSR-based OPM formulation in the NMPC fashion, further proposing a parameterized control policy. We then reformulate the OPM problem as a Bayesian inference problem, which is concerned with identifying the control parameters. The problem is solved by the EnKI method. The resulting OPM approach, departing away from conventional gradient-based optimization, offers significant computational efficiency through combining control policy parameterization, Bayesian inference, and Monte Carlo sampling. We extensively evaluate our proposed approach through simulations under various operational scenarios. The results demonstrate a substantial reduction in computation time, exceeding 90\%, especially for a large number of cells and long optimization horizons. Finally, we build a 20-cell BESS prototype and conduct experimental validation. The results illustrate the effectiveness of our approach in reducing the computational load while addressing imperfect output power predictions.

\balance
\bibliographystyle{IEEEtran}
\scriptsize\bibliography{IEEEabrv,Bibliograpgy/BIB}

\end{document}